\documentclass[11pt,a4paper]{amsart}

\usepackage[utf8]{inputenc}
\usepackage{amssymb}
\usepackage[english]{babel}
\usepackage{graphicx}
\usepackage{amsmath,amssymb}
\usepackage{oldgerm}

\usepackage{mathrsfs}
\usepackage[active]{srcltx}
\usepackage{verbatim}
\usepackage{enumitem} 
\usepackage[numbers]{natbib}
\usepackage[toc,page]{appendix}
\usepackage{aliascnt,bbm}
\usepackage{array}
\usepackage[colorlinks = true, citecolor = blue]{hyperref}
\usepackage[textwidth=4cm,textsize=footnotesize]{todonotes}
\usepackage{xargs}
\usepackage{cellspace}
\usepackage{algorithm}
\usepackage{algorithmic}

\usepackage[Symbolsmallscale]{upgreek}

\usepackage{accents}
\usepackage{dsfont}
\usepackage{aliascnt}
\usepackage{cleveref}
\makeatletter
\newtheorem{theorem}{Theorem}
\crefname{theorem}{theorem}{Theorems}
\Crefname{Theorem}{Theorem}{Theorems}

\newaliascnt{lemma}{theorem}
\newtheorem{lemma}[lemma]{Lemma}
\aliascntresetthe{lemma}
\crefname{lemma}{lemma}{lemmas}
\Crefname{Lemma}{Lemma}{Lemmas}

\newaliascnt{corollary}{theorem}

\aliascntresetthe{corollary}
\crefname{corollary}{corollary}{corollaries}
\Crefname{Corollary}{Corollary}{Corollaries}

\newaliascnt{proposition}{theorem}
\newtheorem{proposition}[proposition]{Proposition}
\aliascntresetthe{proposition}
\crefname{proposition}{proposition}{propositions}
\Crefname{Proposition}{Proposition}{Propositions}

\newaliascnt{definition}{theorem}

\aliascntresetthe{definition}
\crefname{definition}{definition}{definitions}
\Crefname{Definition}{Definition}{Definitions}

\newaliascnt{remark}{theorem}

\aliascntresetthe{remark}
\crefname{remark}{remark}{remarks}
\Crefname{Remark}{Remark}{Remarks}

\crefname{example}{example}{examples}
\Crefname{Example}{Example}{Examples}

\crefname{figure}{figure}{figures}
\Crefname{Figure}{Figure}{Figures}

\newtheorem{assumption}{\textbf{H}\hspace{-3pt}}
\Crefname{assumption}{\textbf{H}\hspace{-3pt}}{\textbf{H}\hspace{-3pt}}
\crefname{assumption}{\textbf{H}}{\textbf{H}}

\Crefname{assumptionG}{\textbf{G}\hspace{-3pt}}{\textbf{G}\hspace{-3pt}}
\crefname{assumptionG}{\textbf{G}}{\textbf{G}}

\makeatother

\def\Abor{\mathsf{A}}

\newcommandx{\functionspace}[2][1=+]{\mathbb{M}_{#1}(#2)}

\newcommand{\argmax}{\operatorname*{arg\,max}}
\newcommand{\argmin}{\operatorname*{arg\,min}}

\newcommandx{\VarDeux}[3][3=]{\operatorname{Var}^{#3}_{#1}\left\{#2 \right\}}

\newcommand{\1}{\mathbbm{1}}
\newcommand{\indiD}[1]{\mathbbm{1}_{\{#1\}}}
\newcommand{\B}{\mathcal{B}}
\newcommand{\borelSet}{\mathcal{B}}

\newcommand{\bx}{\boldsymbol{x}}

\newcommand{\bw}{\boldsymbol{w}}

\newcommand{\LeftEqNo}{\let\veqno\@@leqno}

\newcommand{\N}{\ensuremath{\mathbb{N}}}

\def\rmq{\mathrm{q}}

\newcommand{\abs}[1]{\left\vert #1 \right\vert}
\newcommand{\tvnorm}[1]{\| #1 \|_{\mathrm{TV}}}

\newcommandx{\Vnorm}[2][1=V]{\| #2 \|_{#1}}
\newcommandx{\VnormEq}[2][1=V]{\left\| #2 \right\|_{#1}}
\newcommandx{\norm}[3][1=,2=]{\ifthenelse{\equal{#1}{}}{
    \ifthenelse{\equal{#2}{}} {\left\Vert #3 \right\Vert} {\left\Vert
        #3 \right\Vert_{#2}} } {\ifthenelse{\equal{#2}{}}{\left\Vert
        #3 \right\Vert^{#1}_{#2}}{\left\Vert #3 \right\Vert^{#1}_{#2}}
  }}
\newcommandx{\normLigne}[3][1=,2=]{\ifthenelse{\equal{#1}{}}{
    \ifthenelse{\equal{#2}{}} {\Vert #3 \Vert} {\Vert
        #3 \Vert_{#2}} } {\ifthenelse{\equal{#2}{}}{\Vert
        #3 \Vert^{#1}_{#2}}{\Vert #3 \Vert^{#1}_{#2}}
  }}
\newcommandx{\normUn}[2][1=]{\ifthenelse{\equal{#1}{}}{
\left\Vert #2 \right\Vert_{1}} {\left\Vert #2 \right\Vert^{#1}_{1}}}
  
  \newcommand{\parenthese}[1]{\left(#1 \right)}
  \newcommand{\parentheseLigne}[1]{(#1 )}
  \newcommand{\parentheseDeux}[1]{\left[ #1 \right]}
  \newcommand{\defEns}[1]{\left\lbrace #1 \right\rbrace }

\newcommandx{\ps}[3][1=]{\ifthenelse{\equal{#1}{}}{\left\langle#2,#3
    \right\rangle}{\left\langle#2,#3 \right\rangle_{#1}}}
\newcommandx{\psLigne}[3][1=]{\ifthenelse{\equal{#1}{}}{\langle#2,#3
    \rangle}{\langle#2,#3 \rangle_{#1}}}

\newcommandx\probaMarkovTilde[2][2=]
{\ifthenelse{\equal{#2}{}}{{\widetilde{\mathbb{P}}_{#1}}}{\widetilde{\mathbb{P}}_{#1}\left[ #2\right]}}

\newcommand{\plusinfty}{+\infty}

\newcounter{hypoconbis}
\newcounter{saveconbis}
\newcommand\debutH{\begin{list}
{\textbf{H\arabic{hypoconbis}}}{\usecounter{hypoconbis}}\setcounter{hypoconbis}{\value{saveconbis}}}
\newcommand\finH{\end{list}\setcounter{saveconbis}{\value{hypoconbis}}}

\def\ie{i.e.}
\def\eqsp{\;}

\newcommand{\ocint}[1]{\left(#1\right]}
\newcommand{\ooint}[1]{\left(#1\right)}

\newcommand{\indi}[1]{\1_{#1}}
\newcommandx{\weight}[2][2=n]{\omega_{#1,#2}}

\newcommand{\boule}[2]{\operatorname{B}(#1,#2)}

\def\MAP{\operatorname{MAP}}

\def\rmd{\mathrm{d}}
\newcommandx\sequence[3][2=,3=]
{\ifthenelse{\equal{#3}{}}{\ensuremath{\{ #1_{#2}\}}}{\ensuremath{\{ #1_{#2}, \eqsp #2 \in #3 \}}}}
\newcommandx{\sequencen}[2][2=n\in\N]{\ensuremath{\{ #1, \eqsp #2 \}}}
\newcommandx\sequenceDouble[4][3=,4=]
{\ifthenelse{\equal{#3}{}}{\ensuremath{\{ (#1_{#3},#2_{#3}) \}}}{\ensuremath{\{  (#1_{#3},#2_{#3}), \eqsp #3 \in #4 \}}}}
\newcommandx{\sequencenDouble}[3][3=n\in\N]{\ensuremath{\{ (#1_{n},#2_{n}), \eqsp #3 \}}}
\newcommand{\wrt}{w.r.t.}

\def\iid{i.i.d.}
\def\rme{\mathrm{e}}

\def\rset{\mathbb{R}}

\def\nset{\mathbb{N}}

\def\Lip{\operatorname{Lip}}

\newcommandx{\CPE}[3][1=]{{\mathbb E}^{#3}_{#1}\left[#2 \right]} 
\newcommandx{\CPVar}[3][1=]{\mathrm{Var}^{#3}_{#1}\left\{ #2 \right\}}
\newcommand{\CPP}[3][]
{\ifthenelse{\equal{#1}{}}{{\mathbb P}\left(\left. #2 \, \right| #3 \right)}{{\mathbb P}_{#1}\left(\left. #2 \, \right | #3 \right)}}

\newcommandx{\osc}[2][1=]{\mathrm{osc}_{#1}(#2)}

\newcommand{\chunk}[4][]%
{\ifthenelse{\equal{#1}{}}{\ensuremath{{#2}_{#3:#4}}}{\ensuremath{#2^#1}_{#3:#4}}
}

\def\rmv{\mathrm{v}}

\def\borelean{\mathrm{A}}

\def\convSym{\mathrm{c}}
\def\StSym{\mathrm{s}}

\def\XLang{\mathbf{X}}
\def\XE{X}
\def\ZE{Z}

\def\chull{\operatorname{co}}

\def\rmM{\mathrm{M}}

\def\bfPhi{\mathbf{\Phi}}
\def\Phibf{\mathbf{\Phi}}

\def\xstar{x^{\star}}

\def\bdrift{\beta}

\def\Fsmall{\omega}

\def\RSt{\Rdrift_{\StSym}}
\def\constSt{m}
\def\RStV{\Rdrift_{\StSym}}
\def\constStV{\constSt}

\def\Rdrift{\mathrm{R}}

\def\gaStep{\gamma}

\def\RKer{R}

\def\LL{L_f}

\def\rhoUl{\eta_{\convSym}}
\def\RUl{\mathrm{R}_{\convSym}}

\def\bargaStep{\bar{\gaStep}}

\newcommand{\ball}[2]{\mathrm{B}(#1,#2)}

\def\KOne{\mathrm{a}_{\convSym}}
\def\bOne{\bdrift_{\convSym}}

\def\lsc{l.s.c}
\def\fU{f}
\def\Ul{U^{\lambdaMY}}
\def\pil{\pi^{\lambdaMY}}

\def\gU{g}

\def\gconv{\mathrm{g}}
\def\gconvl{\gconv^{\lambdaMY}}
\def\gUl{\gU^{\lambdaMY}}
\def\lambdaMY{\lambda}
\def\convSet{\mathcal{K}}
\def\indiK{\iota}
\def\prox{\operatorname{prox}}
\newcommandx{\proj}[2]{\ensuremath{\operatorname{proj}_{#2}\left(  #1\right)}}
\newcommandx{\projK}[1]{\ensuremath{\operatorname{proj}_{\convSet}\left(  #1\right)}}

\def\proxgl{\prox_{\gconv}^{\lambdaMY}}
\def\proxgul{\prox_{\gU}^{\lambdaMY}}

\def\rmM{\mathrm{M}}
\def\Uml{U^{\lambdaMY}}
\def\XEM{\XE^{\rmM}}
\def\XEA{\tilde{\XE}^{\lambdaMY}}
\def\piml{\pi^{\lambdaMY}}

\def\xgconv{x_{\gconv}}
\def\Rgconv{R_{\gconv}}
\def\rhogconv{\rho_{\gconv}}
\def\rhoUconv{\rho_{\gU}}
\def\xUconv{x_{\gU}}

\def\vol{\operatorname{Vol}}

\def\ImpSamp{\mathrm{S}}
\def\frm{h}

\def\LUl{L}

\def\ratioEst{\hat{B}}

\def\homogeneousFactorD{\lambda}

\def\Gdiff{\bar{g}}

\def\LLf{L_f}
\def\VSIAM{V_{\convSym}}
\def\lambdaSIAM{\varrho_{\convSym}}
\def\KOneD{\alpha_{\convSym}}
\def\bSIAM{b_{\convSym}}

\def\kappaSIAM{\kappa}
\def\pilgaStep{\pil_{\gaStep}}

\def\barU{\bar{U}}
\def\barpi{\bar{\pi}}

\def\Setharm{\mathsf{A}}

\def\hpdz{\mathcal{C}}
\def\hpd{\mathcal{C}^{\star}}
\def\gconvD{\mathrm{h}}
\def\proxglD{\prox_{\gconvD}^{\lambdaMY}}

\title[Efficient Bayesian computation by Proximal MCMC]{Efficient Bayesian computation by proximal Markov chain Monte Carlo: when Langevin meets Moreau}

\author[A. Durmus]{Alain Durmus$^{1}$}
\address{LTCI, Telecom ParisTech \& CNRS, 46 rue Barrault, 75634 Paris Cedex 13, France}
\email{alain.durmus@telecom-paristech.fr}

\author[E. Moulines]{Eric Moulines$^2$}
\address{Centre de Math\'ematiques Appliqu\'ees, UMR 7641, Ecole Polytechnique, France}
\email{eric.moulines@polytechnique.edu}

\author[M. Pereyra]{Marcelo Pereyra$^3$}
\address{Department of Mathematics,
University of Bristol,
University Walk, Clifton,
Bristol BS8 1TW, U.K.}
\email{marcelo.pereyra@bristol.ac.uk}

\begin{document}

\maketitle

\smallskip


\begin{abstract}
Modern imaging methods rely strongly on Bayesian inference techniques to solve challenging imaging problems. Currently, the predominant Bayesian computation approach is convex optimisation, which scales very efficiently to high dimensional image models and delivers accurate point estimation results. However, in order to perform more complex analyses, for example image uncertainty quantification or model selection, it is necessary to use more computationally intensive Bayesian computation techniques such as Markov chain Monte Carlo methods. This paper presents a new and highly efficient Markov chain Monte Carlo methodology to perform Bayesian computation for high dimensional models that are log-concave and non-smooth, a class of models that is central in imaging sciences. The methodology is based on a regularised unadjusted Langevin algorithm that exploits tools from convex analysis, namely Moreau-Yoshida envelopes and proximal operators, to construct Markov chains with favourable convergence properties. In addition to scaling efficiently to high dimensions, the method is straightforward to apply to models that are currently solved by using proximal optimisation algorithms. We provide a detailed theoretical analysis of the proposed methodology, including asymptotic and non-asymptotic convergence results with easily verifiable conditions, and explicit bounds on the convergence rates. The proposed methodology is demonstrated with four experiments related to image deconvolution and tomographic reconstruction with total-variation and $\ell_1$ priors, where we conduct a range of challenging Bayesian analyses related to uncertainty quantification, hypothesis testing, and model selection in the absence of ground truth.
\end{abstract}

\def\sectionautorefname{Section}
\def\subsectionautorefname{Section}
\def\subsubsectionautorefname{Section}
\def\corollaryautorefname{Corollary}


\section{Introduction}
Image estimation problems are ubiquitous in science and engineering. For example,  problems related to image denoising \cite{Lebrun:2013}, deconvolution \cite{Bonettini2013}, compressive sensing reconstruction \cite{Donoho2006}, super-resolution \cite{Veniamin2016}, tomographic reconstruction \cite{Lustig:2007}, inpainting \cite{Chan2011}, source separation \cite{Zhengming:2012}, fusion \cite{Haro:2015}, and phase retrieval \cite{Candes:2013}. The development of new theory, methodology, and algorithms for imaging problems is a focus of significant research efforts. Particularly, convex imaging problems have received a lot of attention lately, leading to major developments in this area.

Most recent works in the imaging literature adopt formal mathematical approaches to analyse problems, derive solutions, and study the underpinning algorithms. There are several mathematical frameworks available to solve imaging problems \cite{somersalo:2005}. In particular, many modern methods are formulated in the Bayesian statistical framework, which relies on statistical models to represent the data observation process and the prior knowledge available, and then derives solutions by using inference techniques rooted in Bayesian decision theory \cite{somersalo:2005}. 

There are currently two main approaches in Bayesian imaging methodology. The predominant approach is to use a convex formulation of the estimation problem and postulate a prior distribution that is log-concave. This leads to a posterior distribution that is also log-concave, and where maximum-a-posteriori (MAP) estimation can be computed efficiently by using high dimensional convex optimisation algorithms \cite{Green2015}. In addition to scaling well to large settings, convex optimisation algorithms have two additional advantages that are important for practical Bayesian computation: they are well understood theoretically and their conditions for convergence are clear and simple to check; and the main algorithms are general and can be applied similarly to wide range of problems. However, as we will discuss later, convex optimisation on its own cannot deliver basic aspects of the Bayesian paradigm, and struggles to support the complex statistical analyses that are inherent to modern scientific reasoning and decision-making, such as uncertainty quantification and model comparison analyses.

The second main approach in Bayesian imaging methodology is based on stochastic simulation algorithms, namely Markov chain Monte Carlo (MCMC) algorithms. Such methods, which were already actively studied over two decades ago, have regained significant attention lately because of their capacity to address very challenging imaging problems that are beyond the scope of optimisation-based techniques \cite{Pereyra2016}. Additionally to complex models such as hierarchical or empirical Bayesian models, MCMC methods also enable advanced analyses such as hypotheses test and model selection. Unfortunately, despite great progress in high dimensional MCMC methodology, solving imaging problems by stochastic simulation remains too expensive for applications involving moderate or large datasets. Another drawback of existing MCMC methods is that the conditions for their convergence are often significantly more difficult to check than those of optimisation schemes. As a result, most practitioners only assess convergence empirically. It is worth mentioning that some of these limitations can be partially mitigated by resorting to variational Bayes or message passing approximations, which are generally significantly more computationally efficient than stochastic simulation. Unfortunately, such approximations are available only for specific models, and we currently have little theory to analyse the approximation error involved. Similarly, it is generally difficult to provide convergence guarantees for the related algorithms, which often suffer from local convergence issues. Observe that this is in sharp contrast with the convex optimisation approach, which despite its clear limitations, is general and well understood theoretically.

In summary, convex optimisation and MCMC methods have complementary strengths and weaknesses related to their computational efficiency, theoretical underpinning, and the inferences they can support. As a result, it is increasingly acknowledged that the two methodologies should be used together. In this view, the future imaging methodological toolbox should provide a flexible framework where it is possible to perform very efficiently a first analysis of a full dataset by using convex optimisation algorithms, followed by in-depth analyses by MCMC simulation for specific data (e.g., particular data that will be used as evidence to support a hypothesis or a decision). Also, in this framework practitioners should be able to use MCMC algorithms to perform preliminary analyses, which then set the basis for a full scale analysis with convex optimisation techniques. These could be, for example, exploratory analyses with selected data aimed at calibrating the model or performing Bayesian model selection, and benchmarking analyses to assess efficient approximations (e.g., optimisation-based approximate confidence intervals \cite{Pereyra2016b}). Unfortunately, it is currently difficult to use optimisation and MCMC methodologies in this complementary manner because optimisation methods use predominantly non-conjugate priors that are not smooth, such as priors involving the $\ell_1$ or the total-variation noms, whereas MCMC methods are mainly restricted to models with priors that are either conjugate to the likelihood function, or that are smooth with Lipchitz gradients (the latter enables efficient high dimensional MCMC algorithms such as the Metropolis-adjusted Langevin algorithm or Hamiltonian Monte Carlo \cite{Pereyra2016}).

Proximal MCMC algorithms, proposed recently in \cite{Pereyra2015}, are an important first step towards bridging this methodological gap between convex optimisation and stochastic simulation. Unlike conventional high dimensional MCMC algorithms that use gradient mappings and require Lipchitz differentiability, proximal MCMC algorithms draw their efficiency from convex analysis, namely proximal mappings and Moreau-Yoshida envelopes. This allows MCMC-based Bayesian computation for precisely the type of models that are solved by convex optimisation (i.e., high dimensional models that are log-concave but not smooth), which in turn enables advanced Bayesian analyses for these models (e.g., see \cite{Pereyra2016b,Atchade2016} for applications of proximal MCMC to Bayesian uncertainty quantification and sparse regression). However, the proximal MCMC algorithms presented in \cite{Pereyra2015} have three shortcomings that limit their impact in imaging sciences, and which this paper seeks to address. First, the conditions that guarantee the convergence of the algorithms are difficult to check in practice. Second, the algorithms assume that it is possible to compute the proximal mapping of the log-posterior distribution; in practice however this mapping is often approximated by using a forward-backward splitting scheme. Third, the algorithms rely on a Metropolis-Hastings correction step to remove the asymptotic bias introduced by the approximations and to guarantee that the Markov chains target the desired posterior distribution. Unfortunately, this correction step can degrade significantly the efficiency of the algorithms (i.e., the asymptotic bias is removed at the expense of a potentially significant increase in estimation variance and some additional bias from the Markov chain's transient or burn-in regime).

This paper presents a new and significantly better proximal MCMC methodology that address all the issues of the original proximal algorithms discussed above. This new methodology is highly computationally efficient and general, in that it can be applied straightforwardly to most models currently addressed by convex optimisation (in particular, to any model that can be solved by forward-backward splitting). Moreover, we provide simple theoretical conditions to guarantee the convergence of the Markov chains, as well as bounds on its convergence rate. 

The remainder of the paper is organised as follows: Section \ref{sec:bac} defines notation, introduces the class of models considered, and recalls the Langevin MCMC approach that is the basis of our method. In Section \ref{sec:more-yosida-regul} we present the proposed MCMC method, analyse its theoretical properties in detail, provide practical implementation guidelines, and discuss connections with the original proximal MCMC algorithms described in \cite{Pereyra2015}. Section \ref{sec:experiments} illustrates the
methodology on four experiments related to image deconvolution and tomographic reconstruction with total-variation and $\ell_1$ sparse priors, where we conduct a range of challenging Bayesian analyses related to model comparison and uncertainty quantification. Conclusions and perspectives for future work are reported in Section \ref{sec:conclusion}. Proofs are finally reported in Appendices  \ref{sec:proof-crefpr-meas} and \ref{HME}.


\section{Bayesian analysis and computation}\label{sec:bac}
\subsection{Notations and Conventions}
Denote by $\mathcal{B}(\rset^d)$ the Borel $\sigma$-field of $\rset^d$. For all $\Abor \in \mathcal{B}(\rset^d)$, denote by $\vol(\Abor)$ its Lebesgue measure. Denote by $\functionspace[]{\rset^d}$ the set of all Borel measurable functions on $\rset^d$ and for $f \in \functionspace[]{\rset^d}$, $\Vnorm[\infty]{f}= \sup_{x \in \rset^d} \abs{f(x)}$. For $\mu$ a probability measure on $(\rset^d, \mathcal{B}(\rset^d))$ and $f \in \functionspace[]{\rset^d}$ a $\mu$-integrable function, denote by $\mu(f)$ the integral of $f$ \wrt~$\mu$.
For two probability measures $\mu$ and $\nu$ on $(\rset^d, \mathcal{B}(\rset^d))$, the total variation norm of $\mu$ and $\nu$ is defined as
\begin{equation*}
  \tvnorm{\mu-\nu} = \sup_{f \in \functionspace[]{\rset^d}, \Vnorm[\infty]{f} \leq 1}  \abs{\int_{\rset^d } f(x) \rmd \mu (x) - \int_{\rset^d}  f(x) \rmd \nu (x)} \eqsp
\end{equation*}
 Let $f : \rset^d \to
\ocint{-\infty , \plusinfty}$.  If $f$ is a Lipschitz function, namely
there exists $C \geq 0$ such that for all $x,y \in \rset^d$, $\abs{f(x) - f(y)} \leq C \norm{x-y}$, then denote
\begin{equation*}
\norm{f}_{\Lip} =
\inf \{ \abs{f(x) - f(y)} \norm[-1]{x-y} \ | \ x,y \in \rset^d , x
\not = y \}  \eqsp.
\end{equation*}
Denote for all $M \in \rset$, $\{ f \leq M \} = \{ z \in
\rset^d \ | \ f(z) \leq M \}$. $f$ is said to be lower semicontinuous
if for all $M \in \rset$, $\{ f \leq M \}$ is a closed subset of
$\rset^d$.
For $k \geq 0$, denote by $C^k(\rset^d)$, the set of
k-times continuously differentiable functions. For $f \in
C^1(\rset^d)$, denote by $\nabla f$ the gradient of $f$.
%
%
 Denote for all $\rmq
\geq 1$, the $\ell_{\rmq}$ norm $\norm[][\rmq]{\cdot}$ on $\rset^d$ by
for all $x \in \rset^d$, $\norm[][\rmq]{x} = (\sum_{i=1}^d
\abs{x_i}^{\rmq})^{1/\rmq}$. Denote by $\norm{\cdot}$ the Euclidian
norm on $\rset^d$.  For all $x \in \rset^d$ and $M >0$, denote by
$\boule{x}{M}$, the ball centered at $x$ of radius $M$. For a closed
convex $\convSet \subset \rset^d$, denote by $\proj{\cdot}{\convSet}$,
the projection onto $\convSet$, and $\indiK_{\convSet}$ the convex
indicator of $\convSet$ defined by $\indiK_{\convSet}(x)= 0$ if $x \in
\convSet$, and $\indiK_{\convSet}(x)= \plusinfty$ otherwise.
 In the sequel, we take the convention that $\inf \emptyset = \infty$,
$1/\infty = 0$ and for $n,p \in \nset$, $n <p$ then $\sum_{p}^n =0$ and $\prod_p
^n = 1$.


\subsection{Imaging inverse problems}
We consider imaging inverse problems where we seek to estimate an unknown image $x \in \mathbb{R}^d$ from an observation $y$, related to $x$ by a forward statistical model with likelihood function $p(y|x)$. Following a Bayesian approach, we use prior knowledge about $x$ to reduce the uncertainty and deliver accurate estimation results \cite{somersalo:2005}. Precisely, we specify a prior distribution $p(x)$ promoting expected properties (e.g., sparsity, piecewise regularity, or smoothness), and combine observed and prior information by using Bayes' theorem, leading to the posterior distribution \cite{cprbayes}
$$
\pi(x) \triangleq p(x|y) = \frac{p(y|x)p(x)}{\int_{\mathbb{R}^d} p(y|x)p(x)\textrm{d}x}\,,
$$
that we henceforth denote as $\pi$, and which models our knowledge about $x$ after observing $y$. In this paper we focus on inverse problems that are convex. We assume that $\pi$ is log-concave, \ie
\begin{eqnarray}\label{posterior}
\pi(x) = \frac{\rme^{-U(x)}}{\int_{\mathbb{R}^d} \rme^{-U(s)} \rmd s \eqsp},
\end{eqnarray}
for some measurable function $U : \rset^{d} \to \ocint{-\infty,\plusinfty}$ satisfying the following condition.
 \begin{assumption}
   \label{assum:form-potential}
   $U = \fU+ \gU$, where $\fU : \rset^d \to \rset$ and $\gU : \rset^d \to \ocint{-\infty,\plusinfty}$ are two lower bounded functions
   satisfying:
   \begin{enumerate}[label=(\roman*)]
   \item $\fU$ is convex, continuously differentiable, and gradient Lipschitz with Lipschitz constant $\LL$, \ie~for all $x,y \in \rset^d$
     \begin{equation}
       \label{eq:gradient-Lip}
       \norm{\nabla \fU(x) - \nabla \fU(y)} \leq \LL \norm{x-y} \eqsp.
     \end{equation}
     \item \label{item:assum-prior} $\gU$ is proper, convex and lower semi continuous (\lsc).
   \end{enumerate}
 \end{assumption}

Notice that the class \eqref{posterior} comprises many important models that are used extensively in modern imaging sciences. Particularly, models of the form $U(x) = \|y-Ax\|^2 /2\sigma^2 + \phi(B x)$ for some linear operators $A$, $B$, and convex regulariser $\phi$ that is typically non-smooth, and which may also encode convex constraints on the parameter space. In such cases $\fU(x) = \|y-Ax\|^2 /2\sigma^2$ and $g(x) = \phi(B x)$ for instance.

When $x$ is high-dimensional, drawing inferences from $\pi$ directly is generally not possible. Instead we use summaries, particularly point estimators, that capture some of the information about $\pi$ that is relevant for the application considered \citep{cprbayes}. In particular, modern statistical imaging methodology relies strongly on the maximum-a-posteriori (MAP) estimator defined by: 
\begin{eqnarray}\label{map}
\begin{split}
\hat{x}_{\MAP} = \argmax_{x \in \mathbb{R}^d} \pi(x) = \argmin_{x \in \mathbb{R}^d} U(x) \eqsp,
\end{split}
\end{eqnarray}
which can often be computed efficiently, even in very large problems, by using proximal convex optimisation algorithms \citep{pesquet:2011, Parikh2013}. From the practitioner's viewpoint, this is a main advantage \wrt~most other summaries that require high-dimensional integration \wrt~$\pi$, which is generally significantly more computationally expensive \citep{pereyra:2016}.

However, in its raw form, mathematical imaging based on optimisation struggles to support the complex statistical analyses that are inherent to modern scientific reasoning. For example, such methods are typically unable to assess the uncertainty in the solutions delivered, support uncertainty quantification and decision-making procedures (e.g. hypothesis tests). Similarly, they have difficulty checking and comparing alternative mathematical models intrinsically (i.e., without ground truth available). To perform such advanced (often Bayesian) analyses and deliver the full richness of the statistical paradigm it is necessary to use Monte Carlo stochastic simulation algorithms \citep{Green2015}. 

As mentioned previously, the high-dimensionality and the lack of smoothness of $\pi$ pose important challenges from a Bayesian computation viewpoint. This paper presents a new MCMC methodology to tackle this problem. The proposed methodology is general, robust, theoretically sound, and computationally efficient, and can be applied straightforwardly to any model satisfying \eqref{posterior} that can be addressed by using proximal convex optimisation (particularly by using the gradient of $\fU$ and the proximal operator of $\gU$, similarly to forward-backward splitting algorithms).

 \subsection{Bayesian computation: unadjusted and Metropolis-adjusted Langevin algorithms}
 The MCMC method proposed in this paper is derived from the
 discretization of overdamped Langevin diffusions.  
 Let $\barU : \rset^d \to \rset$ be a continuously differentiable function and consider the
 Langevin stochastic differential equations (SDE)  given by
 \begin{equation}
   \label{eq:langevin-1}
   \rmd \XLang_t = -\nabla \barU(\XLang_t) \rmd t + \sqrt{2} \rmd B^d_t \eqsp,
 \end{equation}
 where $(B_t^d)_{t \geq 0}$ is a $d$-dimensional Brownian
 motion. Under additional mild assumptions, this equation has a unique
 strong solution. In addition if $\int_{\rset} \rme^{-\barU(x)} \rmd x < \infty$,  then $\barpi(x) \propto \rme^{-\barU(x)}$ is the unique invariant
 distribution of the semi-group associated with the
 Langevin SDE, see \cite{khasminskii:1960}. Consequently, if we could solve \eqref{eq:langevin-1} and let $t \rightarrow \infty$, this would provide samples from $\bar{\pi}$ useful for Bayesian computation. Since it is
 possible to analytically solve \eqref{eq:langevin-1} only in very
 specific cases, we consider a discrete-time Euler-Maruyama
 approximation and obtain the following Markov chain $(\XE_k)_{k\geq
   0}$: for all $k \geq 0$
 \begin{equation}
 \label{eq:definition-Euler}
\textrm{ULA}:\, \XE_{k+1} = \XE_k - \gaStep \nabla \barU(\XE_k) + \sqrt{2 \gaStep} \ZE_{k+1} \eqsp,
 \end{equation}
 where $\gaStep >0$ is a given
 step size and $(\ZE_{k})_{k\geq 1}$ is a sequence of
 \iid~$d$-dimensional standard Gaussian random variables. This scheme
 has been first introduced in molecular dynamics by \cite{ermak:1975}
 and \cite{parisi:1981}, and then popularized in the machine learning community by
 \cite{grenander:1983}, \cite{grenander:miller:1994} and in computational statistics by
 \cite{neal:1992} and \cite{roberts:tweedie-Langevin:1996}. Following  \cite{roberts:tweedie-Langevin:1996}, this algorithm is referred to as the Unadjusted Langevin Algorithm (ULA).

 In Bayesian computation, the samples $(\XE_k)_{k\geq 0}$ generated by
 ULA \eqref{eq:definition-Euler} are used to estimate probabilities and expectations
 w.r.t. $\barpi$. This scheme has attracted significant attention in the late, particularly for high-dimensional problems were most Monte Carlo methods struggle. 
Theory for ULA advanced significantly recently with the development of non-asymptotic bounds in total variation
distance between $\barpi$ and the marginal laws of the Markov chain
 $(\XE_k)_{k \geq 0}$ defined by ULA \cite{dalalyan:2014,durmus:moulines:2015}, with explicit dependence on the stepsize
 $\gaStep$ and the dimension $d$ (see
 \Cref{ssec:convergence_analysis}). These new theoretical results are
 important because they provide estimation accuracy guarantees for ULA, as well as
 valuable new insights into the convergence properties of the
 algorithm. In particular, they establish that if $\barU$ is convex and
 gradient Lipchitz, then ULA's convergence properties deteriorate at
 most polynomially as $d$ increases.
 Remarkably, if in addition $\barU$ is strongly
 convex, then it deteriorates at most linearly with $d$, confirming
 the empirical evidence that ULA is a highly computationally efficient method to sample in
 high-dimensional settings.

It is worth emphasising at this point that this deep understanding of ULA is very recent. Indeed, without a proper theoretical underpinning, ULA has been traditionally regarded as unreliable and rarely applied directly in statistics or statistical image processing. Instead, most applications reported in the literature adopt a safe approach and complement ULA with a Metropolis-Hastings correction step targeting $\barpi$, as recommended by \cite{rossky:doll:friedman:1978} and \cite{roberts:tweedie-Langevin:1996}. This correction guarantees that the resulting Metropolis Adjusted Langevin Algorithm (MALA) generates a reversible Markov chain with respect to $\barpi$, and therefore eliminates the asymptotic bias. And perhaps more importantly, it places ULA within the sound theoretical framework of Metropolis-Hasting algorithms. For sufficiently smooth densities MALA inherits the good convergence properties of ULA and scales efficiently to high-dimensional settings \citep{roberts:tweedie-Langevin:1996}. 

Unfortunately, neither ULA nor MALA are well defined for non-smooth target
densities, which strongly limits their
application to modern mathematical imaging problems. In fact, both theory and experimental
evidence show that ULA and MALA often run into difficulties if $\pi$ is not sufficiently regular. For example, when $\nabla\log\pi$ is not Lipchitz  continuous ULA is generally explosive and MALA is not geometrically ergodic (see \cite[Figure
2]{roberts:tweedie-Langevin:1996, pereyra:2015}). Similarly, when $\nabla\log\pi$  is subdifferentiable and therefore, at least from a purely algorithmic viewpoint, the algorithms could still be applied, the theory underpinning the ULA and MALA collapses and even the convergence of the time-continuous Langevin diffusion driving the algorithms becomes unclear.
Moreover, many applications involve constraints on the parameter space and then $\pi$ is supported only on a bounded convex set $\convSet$. In
such case, $\nabla\log\pi$ is bounded on $\convSet$ and infinite or not defined
outside $\convSet$. Then it is not possible to use ULA, and MALA
typically behaves very poorly (the algorithm gets ``stuck'' whenever
the proposal drives the Markov chain outside $\convSet$). Following a proximal MCMC approach \cite{pereyra:2015}, in the following section we present a new ULA that exploits tools from convex calculus and proximal optimisation to address these issues, and sample efficiently from high-dimensional log-concave densities of the form H\ref{assum:form-potential} that are beyond the scope of conventional ULAs and MALAs.




\section{Proximal MCMC: Moreau-Yosida regularised Unadjusted Langevin Algorithm}
\label{sec:more-yosida-regul}
\subsection{Proposed method}
\label{ssec:MYULAl}
A central idea in this work is to replace the non-smooth potential $U$ with a carefully designed smooth approximation $U^{\lambda}$ which, by construction, has the following two key properties: 1) its Euler-Maruyama discrete-time approximations are always stable and have favourable convergence properties, and 2) we can make  $\pi^\lambda\propto\rme^{-U^{\lambda}}$ arbitrarily close to $\pi$ by adjusting an approximation parameter $\lambda >0$. 

In a manner akin to \cite{pereyra:2015}, we define such approximations by using Moreau-Yosida envelopes
\cite{Combettes2011} which we recall below. Let $\gconv :
\rset^d \to \ocint{-\infty,+\infty}$ be a \lsc~convex function and
$\lambdaMY >0$. The $\lambda$-Moreau-Yosida envelope of $\gconv$ is a carefully regularised approximation of $g$ given by
\begin{equation}
\label{eq:id-MY-env}
\gconv^{\lambdaMY}(x) = \min_{y \in \rset^d} \defEns{\gconv(y) + (2 \lambdaMY)^{-1}\norm[2]{x-y}} \eqsp,
\end{equation}
where $\lambda$ is a regularisation parameter that controls a trade-off between the regularity properties of $\gconv^{\lambdaMY}$ and the approximation error involved. Remarkably, by \cite[Example 10.32, Theorem 9.18]{rockafellar:wets:1998}, the approximation $\gconvl$ inherits the convexity of $g$ and is always continuously differentiable, even if $g$ is not. In fact, $\gconvl$ is gradient Lipshitz  \cite[Proposition 12.19]{rockafellar:wets:1998}: for all $x,y \in \rset^d$,
\begin{equation}
\label{eq:lip_moreau_yosida}
\norm{\nabla \gconvl(x) - \nabla \gconvl(y)} \leq \lambdaMY^{-1} \norm{x-y} \eqsp.
\end{equation}
The gradient is given by for all $x \in \rset^d$
\begin{equation}
\label{eq:definition-grad-prox}
\nabla \gconvl(x) = \lambdaMY^{-1}\parenthese{x-\proxgl(x)} \eqsp,
\end{equation}
where
\begin{equation}
\label{eq:id-MY-env}
\proxgl(x) = \argmin_{y \in \rset^d} \defEns{\gconv(y) + (2 \lambdaMY)^{-1}\norm[2]{x-y}} \eqsp,
\end{equation}
is the proximal operator of $\gconv$ \cite{Combettes2011}. This operator is used extensively in imaging methods based on convex optimisation, where it is generally computed efficiently by using a specialised algorithm \cite{pesquet:2011,Parikh2013}. Indeed, similarly to gradient mappings, $\proxgl$ also moves points in the direction of the minimum of $g$ (by an amount related to the value of $\lambda$), and has many properties that are useful for devising fixed-point methods \cite{Combettes2011}.

In addition, $\gconvl$ envelopes $\gconv$ from below: for all $x \in \rset^d$, $\gconv^{\lambdaMY}(x) \leq \gconv(x)$, and since for  $0 < \lambdaMY < \lambdaMY'$ and $x,y \in \rset^d$,
$\gconv(y) + (2 \lambdaMY')^{-1}\norm[2]{x-y} \leq \gconv(y) + (2 \lambdaMY)^{-1}\norm[2]{x-y}$, we get that for all $x \in \rset^d$ $\gconv^{\lambdaMY'}(x) \leq \gconv^{\lambdaMY}(x)$. By \cite[Theorem 1.25]{rockafellar:wets:1998}, $\gconvl$ converges pointwise to $\gconv$ as $\lambdaMY$ goes to $0$, \ie\ for all $x \in \rset^d$,
\begin{equation}
\label{eq:limit-d-lambda}
\lim_{\lambdaMY \to 0}\gconvl(x) = \gconv(x) \eqsp.
\end{equation}
Hence, $\gconvl$ provides a convex and smooth approximation to $g$ that we can make arbitrarily close to $\gconv$ by adjusting the value of $\lambda$.



So under \Cref{assum:form-potential}, if  $\gU$ is
not continuously differentiable, but the proximity operator associated with $\gU$ is
available, we can consider sampling algorithms that use the $\lambdaMY$-Moreau-Yosida envelope
$\gUl$ instead of $g$. Here we propose to replace the potential $U$ with the approximation $\Uml:\rset^d \to \rset$ defined for all $x \in \rset^d$ by 
 \begin{equation*}
\Uml(x) = \gUl(x) + \fU(x) \eqsp,
\end{equation*}
which we will use to define a surrogate target density $\pi^\lambda \propto \rme^{-U^\lambda}$. We will see that such approximation is endowed with very useful regularity and approximation accuracy properties.

\Cref{propo:finite-measure-MY} below implies that the probability measure $\piml$ on $\rset^d$,  with density  with respect to the Lebesgue measure, also denoted by $\piml$ and given for all $x \in \rset^d$ by 
$$
\pi^\lambda (x)= \frac{\rme^{-U^\lambda(x)}}{\int_{\rset^d} \rme^{-U^\lambda (s)} \rmd s} \eqsp,
$$
is well defined, log-concave, Lipschitz continuously differentiable, and as close to $\pi$ as required.

\begin{assumption}
  \label{assum:integrabilite}
  Assume that one of these two conditions holds:
  \indent\begin{enumerate}[label=(\roman*)]
  \item\label{assum:integrable_g} $\rme^{-g}$ is integrable with respect to the Lebesgue measure.
  \item\label{assum:lipschitz_g} $g$ is Lipschitz.
  \end{enumerate}
\end{assumption}
\begin{proposition}
  \label{propo:finite-measure-MY}
Assume \Cref{assum:form-potential} and \Cref{assum:integrabilite}. 
\begin{enumerate}[label=\alph*)]
\item  For all $\lambdaMY >0$, $\pi^\lambda$ defines a proper density of a probability measure on $\rset^d$, \ie 
\begin{equation*}
0 < \int_{\rset^d} \rme^{-\Uml(y)}\rmd y < \plusinfty \eqsp.
\end{equation*}
\item For all $\lambdaMY >0$,  $\pi^\lambda$ is log-concave and  continuously differentiable with  
\begin{equation}
\label{eq:definition-grad-prox_U}
   \nabla U^\lambda (x)=-\nabla \log \pi^\lambda(x) =  \nabla f(x) +\lambdaMY^{-1}(x-\prox^\lambda_{ g}(x)) \eqsp.
\end{equation}
In addition, $\nabla \Uml$ is  Lipschitz with  constant $L \leq \LL + \lambda^{-1}$.
\item 
   \label{item:propo:dist_TV_MY_1}
The approximation $\pi^\lambda$ converges to $\pi$ as $\lambda \downarrow 0$ in total variation norm, \ie 
   \begin{equation*}
   \lim_{\lambdaMY \to 0} \tvnorm{\piml-\pi} = 0  \eqsp.
   \end{equation*}
\item 
   \label{item:propo:dist_TV_MY_2}
If \Cref{assum:integrabilite}-\ref{assum:lipschitz_g} then for all $\lambdaMY >0$, 
\begin{equation*}
\tvnorm{\piml-\pi} \leq \lambdaMY \norm[2][\Lip]{\gU} \eqsp.
\end{equation*}
\end{enumerate}
\end{proposition}
\begin{proof}
The proof is postponed to \Cref{sec:proof-crefpr-meas}.
\end{proof}
 Figure \ref{FigMoreauApprox} shows the approximations of two non-smooth densities that satisfy \Cref{assum:form-potential}: 
\begin{enumerate}
\item\label{item:caseLaplace} the Laplace density $\pi(x) = (1/2)\exp\parenthese{\abs{x}}$, for which 
\[
\piml(x) = \frac{\exp\defEns{(\lambdaMY/2-\abs{x})\indiD{\abs{x}
    \geq \lambdaMY} - (x^2/(2\lambdaMY))\indiD{\abs{x} < \lambdaMY} }}{2 \defEns{\rme^{-\lambdaMY/2}+(2 \uppi
  / \lambdaMY)^{1/2}(\Phibf(\lambdaMY^{1/2})-1/2)}}
\eqsp,
\]
where $\Phibf$ is the cumulative function of the standard normal distribution.
\item\label{item:caseUnif}  the uniform density $\pi(x) = (1/2)\exp\parentheseLigne{-\indiK_{[-1,1]}(x)}$, for which 
  \begin{equation*}
    \piml(x) = \defEns{2+\sqrt{2\uppi\lambdaMY}}^{-1}\exp\parentheseDeux{\defEns{-\max(\abs{x}-1,0)}^2/(2 \lambdaMY)} \eqsp.
  \end{equation*}
\end{enumerate}
We observe that the approximations are smooth and converge to $\pi$ as
$\lambda$ decreases, as described by \Cref{propo:finite-measure-MY}.
Also for these two examples, analytic expressions for
$\tvnorm{\pi-\piml}$ can be found, and \Cref{FigMoreauApprox-TV}
shows $\tvnorm{\pi-\piml}$ as a function of $\lambdaMY >0$. Notice that in the case of the Laplace density $\tvnorm{\pi-\piml}$ goes to $0$ quadratically in $\lambdaMY$ as
$\lambdaMY$ goes to $0$, which is faster than the linear bound given in \Cref{propo:finite-measure-MY}-\ref{item:propo:dist_TV_MY_2}. Also note that this bound does not apply to the uniform density, and in this case $\tvnorm{\pi-\piml}$ vanishes at rate $\sqrt{\lambdaMY}$.
\begin{figure}[htbp!]
\begin{minipage}[l2]{0.5\linewidth}
  \centering
  \centerline{\includegraphics[width=7cm]{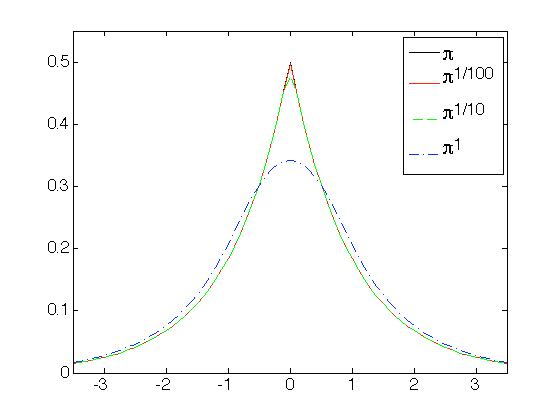}}
  \small{(a) $\pi(x) = \tfrac{1}{2}\rme^{-|x|}$}
\end{minipage}
\begin{minipage}[l2]{0.5\linewidth}
  \centering
  \centerline{\includegraphics[width=7cm]{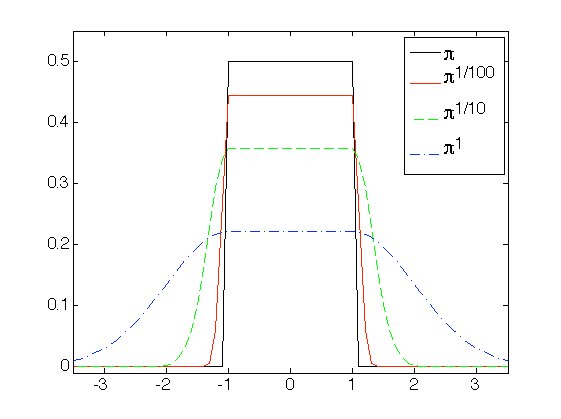}}
  \small{(b) $\pi(x) = \tfrac{1}{2}\indiK_{[-1,1]}(x)$}
\end{minipage}
\caption{\small{Density plots for the Laplace (a) and uniform (b) distributions (solid
black), and their smooth approximations $\pi^\lambda$ for $\lambda = 1, 0.1, 0.01$ (dashed blue and green, and solid red).}} \label{FigMoreauApprox}
\end{figure}

\begin{figure}[htbp!]
\begin{minipage}[l2]{0.5\linewidth}
  \centering
  \centerline{\includegraphics[width=7cm]{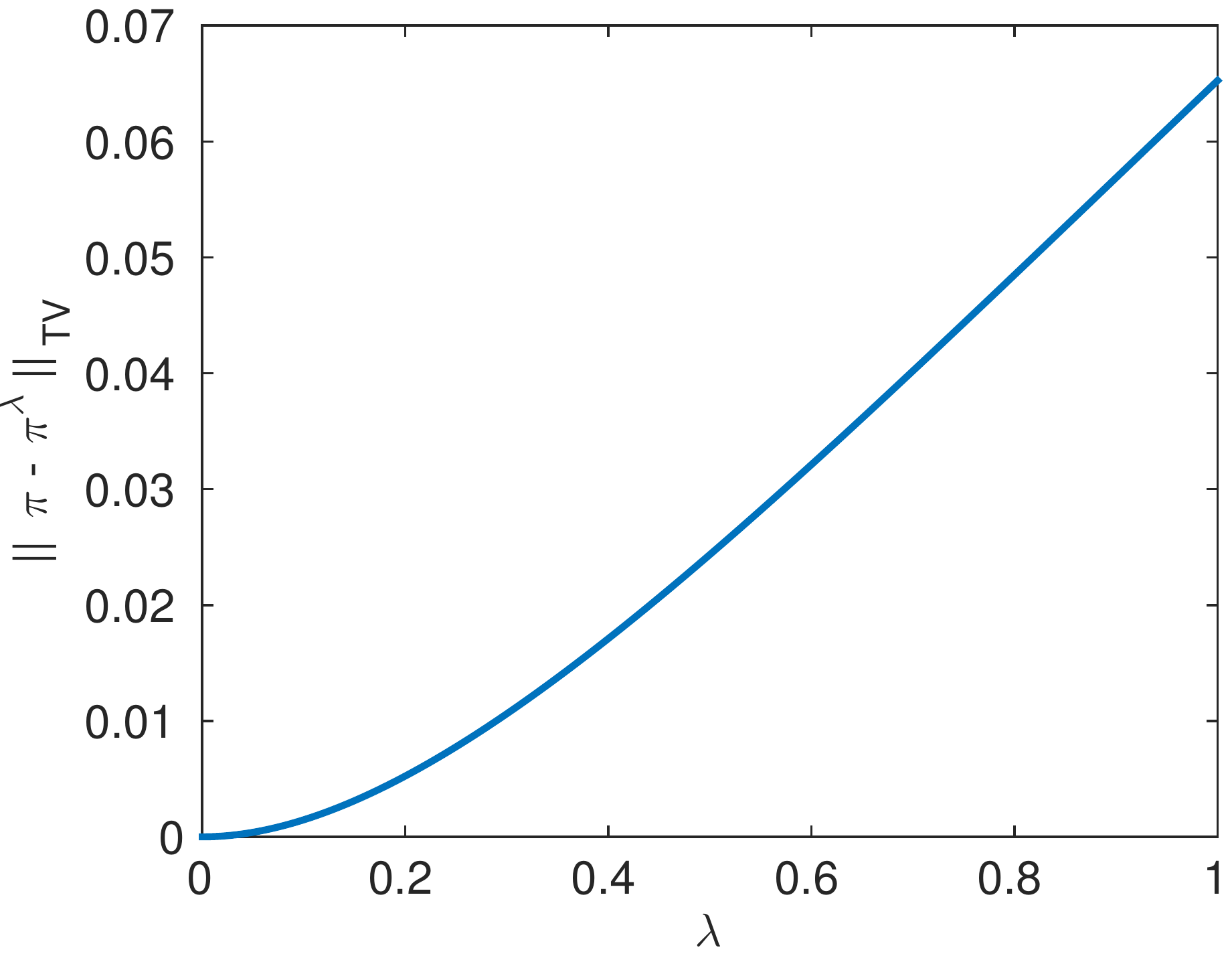}}
  \small{$\pi(x) = \tfrac{1}{2}\rme^{-|x|}$}
\end{minipage}
\begin{minipage}[l2]{0.5\linewidth}
  \centering
  \centerline{\includegraphics[width=7cm]{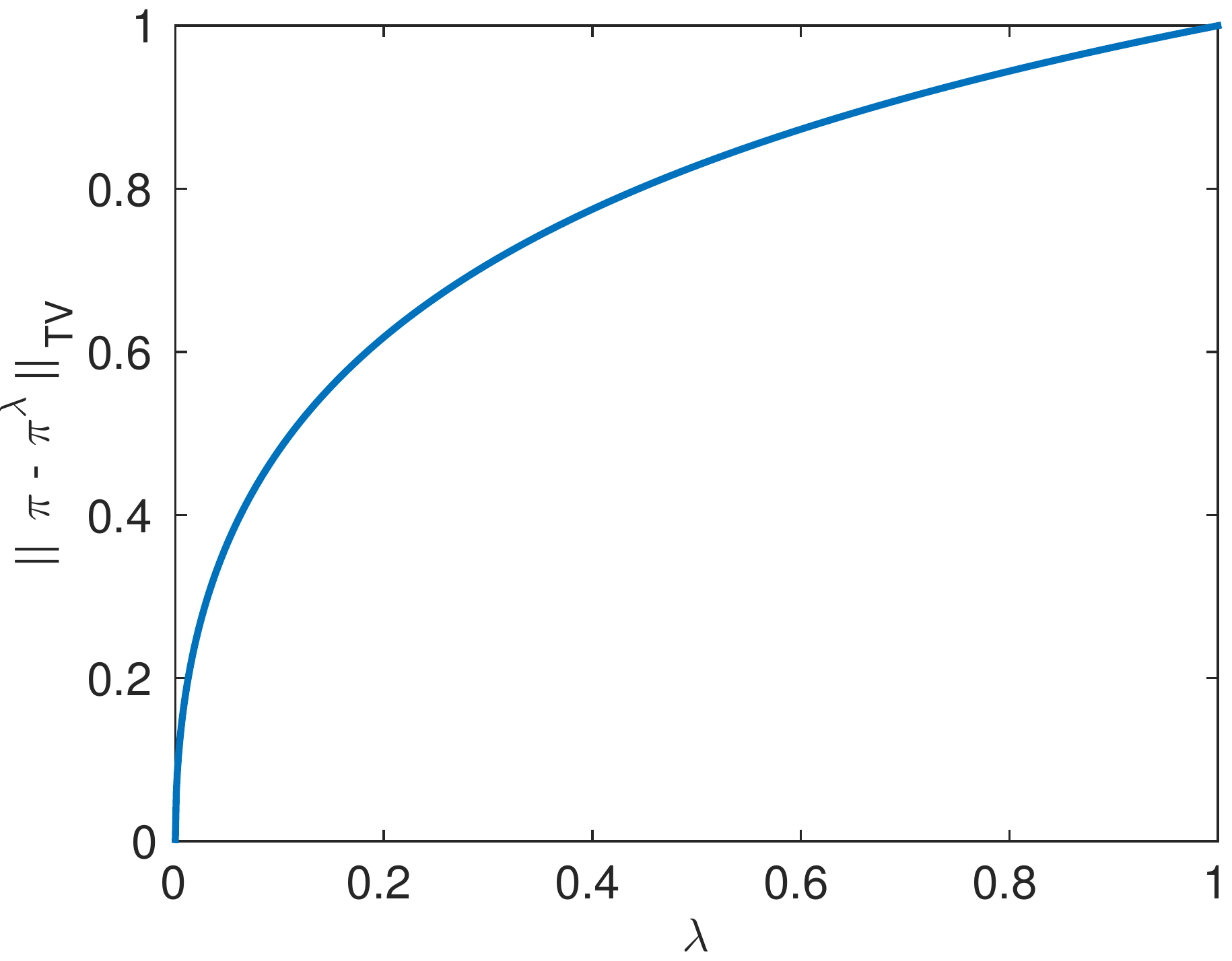}}
  \small{$\pi(x) = \tfrac{1}{2}\indiK_{[-1,1]}(x)$}
\end{minipage}
\caption{\small{Total variation norm between $\pi$ and its smooth approximation $\pil$ as function of $\lambda$.}} \label{FigMoreauApprox-TV}
\end{figure}

We now make two key observations. First, \Cref{propo:finite-measure-MY} shows that $\nabla U^{\lambdaMY}$ is
gradient Lipschitz and therefore it guarantees that the Langevin SDE
constructed with $U^{\lambdaMY}$ converges to $\pi^{\lambda}$ as $t \rightarrow \infty$ (formally, it guarantees that the Langevin SDE
associated with $\pi^{\lambda}$ admits a unique strong solution
$(\mathbf{X}^\lambda_t)_{t \geq 0}$ and $\pi^{\lambda}$ is the unique
stationary distribution of the semigroup). More importantly, as it will be seen below, it 
implies that the ULA chain derived from a Euler-Maruyama discretisation of
this Langevin diffusion will be, by construction, well
behaved and useful for Monte Carlo integration with respect to
$\pi^\lambda$. 

Second, \Cref{propo:finite-measure-MY} also establishes that $\lambda$ controls
the estimation bias involved in performing estimations with
$\pi^\lambda$ as a substitute of $\pi$. This approximation error can
be made arbitrarily small, and is bounded explicitly by $\lambdaMY
\norm[2][\Lip]{\gU}$ when $\gU$ is Lipschitz.

We are now in a position to present the new MCMC methodology proposed in this work, which is essentially an application of ULA to $\pi^{\lambda}$. Precisely, given $\lambdaMY>0$ and a  stepsize $\gaStep >0$, we use an Euler-Maruyama approximation of $(\mathbf{X}^\lambda_t)_{t \geq 0}$, and obtain the following Markov chain $(\XEM_k)_{k\geq 0}$: for all $k \geq 0$
\begin{equation}
  \label{eq:def-MYRULA}
  \textrm{MYULA}: \, \XEM_{k+1} = (1- \tfrac{\gaStep}{\lambda})\XEM_{k} - \gaStep \nabla \fU(\XEM_k)  + \tfrac{\gaStep}{\lambda}\proxgul(\XEM_k) +\sqrt{2 \gaStep} \ZE_{k+1} \eqsp,
\end{equation}
where $\sequence{\ZE}[k][\nset^*]$ is a sequence of \iid\ $d$ dimensional standard Gaussian random variables. This algorithm will be referred to as the \emph{Moreau-Yosida Unadjusted Langevin Algorithm} (MYULA), and is summarised in \Cref{Algo:MYULA} below (see \Cref{guidelines} for guidelines for setting the values of $\gaStep$ and $\lambda$). Note that the stationary distribution of the MYULA sequence $\sequence{\XEM}[k][\nset]$ is different from the  target distribution $\piml$, and depends on the stepsize $\gaStep >0$. Nevertheless, we show in \Cref{ssec:convergence_analysis} that, choosing $\lambda$ and $\gamma$ appropriately, the samples are very close to $\pi$.

Besides, to compute the expectation of a function $\frm : \rset^d \to
\rset$ under $\pi$ from $\{\XEM_k\ ; \ 0 \leq k \leq n\}$,  an
optional importance sampling step might be used to correct the regularization.
This step amounts to approximate $\int_{\rset^d} \frm(x) \pi(x) \rmd x$ by the weighted sum
\begin{equation}
  \label{eq:importance_sampling}
  \ImpSamp_n(\frm) = \sum_{k=0}^n \weight{k}\frm(X_k) \eqsp, \text{ with } \weight{k} = \defEns{\sum_{\ell=0}^n \rme^{\Gdiff^\lambdaMY(\XEM_\ell)}}^{-1} \rme^{\Gdiff^\lambdaMY(\XEM_k)} \eqsp,
\end{equation}
where for all $x \in \rset^d$
\begin{equation*}
  \Gdiff^\lambdaMY(x)=\gUl(x)-\gU(x)=\gU(\proxgul(x)) - \gU(x) +(2\lambdaMY)^{-1}\norm[2]{x- \proxgul(x)} \eqsp.
\end{equation*}

To remove this asymptotic bias, we can add an Hastings-Metropolis step, which will
produce a Markov chain $\sequence{\XEA}[k][\nset]$ which is reversible this time
with respect to $\piml$ and use similarly an importance sampling step
to correct for the bias introduced by smoothing. This algorithm will be called the \emph{Moreau-Yosida Regularized Metropolis-adjusted Langevin Algorithm} (MYMALA). 

The focus of this work is on MYULA without importance sampling or Metropolis-Hastings correction. A study of MYMALA  is currently in progress and will be reported separately.

\begin{algorithm}
\caption{Moreau-Yoshida unadjusted Langevin algorithm (MYULA)}
\label{Algo:MYULA}
\begin{algorithmic}
\STATE \textbf{set} $\XEM_{0} \in \mathbb{R}^d$, $\lambda > 0$, $\gaStep \in (0, \lambda/(\lambda L_f + 1)]$, $n \in \mathbb{N}$
\FOR {$k = 0:n$}
\STATE $\ZE_{k+1} \sim \mathcal{N}(0,\mathbb{I}_d)$
\STATE $\XEM_{k+1} = (1- \tfrac{\gaStep}{\lambda})\XEM_{k} - \gaStep \nabla \fU(\XEM_k)  + \tfrac{\gaStep}{\lambda}\proxgul(\XEM_k) +\sqrt{2 \gaStep} \ZE_{k+1}$
\ENDFOR
\end{algorithmic}
\end{algorithm}

\subsection{Theoretical convergence analysis of MYULA}\label{ssec:convergence_analysis}
In this section we present a detailed theoretical analysis of MYULA implemented with fixed regularization parameter $\lambdaMY>0$ and step-size $\gaStep >0$. We first establish that the chains generated by MYULA converge geometrically fast to an approximation of $\pi$ that is controlled by $\lambda$ and $\gaStep$, and which can be made arbitrarily close to $\pi$. More importantly, we also establish non-asymptotic bounds for the estimation error of MYULA with a finite number of iterations. This enables an analysis of the behaviour of MYULA as the dimensionality of the model increases, as well as deriving practical guidelines for setting $\lambda$ and $\gamma$ for specific models.

First, under \Cref{assum:form-potential}, it has been observed that
$\gUl$ is $\lambdaMY^{-1}$-gradient Lipschitz, which implies that
$\Ul$ is gradient Lipschitz as well: there exists $\LUl \geq 0$ such that for
all $x,y \in \rset^d$, $\norm{\nabla \Ul (x) - \nabla \Ul(y)} \leq
\LUl\norm{x-y}$ and
\begin{equation}
  \label{eq:definition_const_lip_u_lambda}
   \LUl \leq  \LLf + \lambdaMY^{-1} \eqsp.
\end{equation}
Of course, this bound strongly depends on the decomposition of $U$ in a smooth and a non-smooth part, which is arbitrary and therefore
can be pessimistic (for instance, if $U$ is continuously differentiable, $\gU$ can be chosen to be $0$ which implies $\Ul = U$
and $\LUl = \LL$).

We assume first the following assumption on the potential $\Ul$.
\begin{assumption}
  \label{assum:potentialUl}
There exist a minimizer $\xstar$  of $\Ul$, $\rhoUl >0$ and $\RUl \geq 0$ such that  for all $x \in \rset^d$, $\norm{x-\xstar} \geq \RUl$,
\begin{equation}
\label{eq:superexpo_potential}
\Ul(x) - \Ul(\xstar) \geq \rhoUl \norm{x-\xstar} \eqsp.
\end{equation}
\end{assumption}
Note that in fact
\Cref{assum:potentialUl} always holds under \Cref{assum:form-potential} and \Cref{assum:integrabilite}, 
since by \Cref{lem:control-fun-convex-gene} and
\Cref{propo:finite-measure-MY} there exist $C_1,C_2 >0$ such that
$\Ul(x) \geq C_1\norm{x} -C_2$. Therefore, since $\Ul$ is continuous on $\rset^d$, there exists a minimizer $\xstar$ of $\Ul$ and 
\eqref{eq:superexpo_potential} holds with $\rhoUl \leftarrow C_1/2$ and
$\RUl \leftarrow 2(C_2 + \norm{\xstar}+ \Ul(\xstar))/C_1$. However, these
constants are non quantitative, and that is why we introduce \Cref{assum:potentialUl} to derive quantitative bounds.

Consider the Markov kernel $\RKer_{\gaStep}$ associated to the Euler-Maruyama discretization \eqref{eq:def-MYRULA}  given, for all $\Abor \in \B(\rset^d)$ and $x \in \rset^d$ by
\begin{equation}
  \label{eq:definition_R_kernel}
  \RKer_{\gaStep}(x ,\Abor) = (4 \uppi \gaStep)^{-d/2}\int_{\Abor}\exp \parenthese{- (4 \gaStep)^{-1}\norm[2]{y-x+\gaStep \nabla \Ul(x)}} \rmd y \eqsp.
\end{equation}
The sequence $(\XEM_n)_{n \geq 0}$ defined by \eqref{eq:def-MYRULA}  is a homogeneous Markov chain
associated with the Markov kernel $\RKer_{\gaStep}$.

It is easily seen that under \Cref{assum:form-potential}, since $\Ul$ is continuously differentiable, $\RKer_{\gaStep}$
is irreducible with
 respect to the Lebesgue measure,  all compact sets are $1$-small and the kernel is strongly aperiodic.  In addition under \Cref{assum:potentialUl}, since $U$ is also convex then
 \cite[Proposition 13]{durmus:moulines:2015} shows that $\RKer_{\gaStep}$ satisfies a Foster-Lyapunov drift condition, \ie~for all $\bar{\gaStep} \in \ocint{0,L}$, $\gaStep \in \ocint{0,\bar{\gaStep}}$ and for all $x \in \rset^d$,
 \begin{equation*}
   \RKer_{\gaStep} \VSIAM(x) \leq \lambdaSIAM^{\gaStep} \VSIAM(x) + \bSIAM \gaStep \eqsp,
 \end{equation*}
where 
\begin{subequations}
\begin{align}
\VSIAM(x)& = \exp\defEns{(\rhoUl/4)\parenthese{\norm[2]{x-\xstar}+1}^{1/2}} \\
  \lambdaSIAM &= \rme^{-2^{-4} \rhoUl^2(2^{1/2}-1)} \eqsp, \ \KOne = \max(1,2 d / \rhoUl,\RUl) \\
\bSIAM &= \defEns{ (\rhoUl/4)(d+(\rhoUl \bargaStep/4))  - \log( \lambdaSIAM )}  \rme^{\rhoUl (\KOne^2+1)^{1/2}/4 + (\rhoUl \bargaStep/4)(d+(\rhoUl \bargaStep/4))} 
 \eqsp.
\end{align}
\label{eq:convex_drift}
\end{subequations}
By \cite[Theorem 16.0.1]{meyn:tweedie:2009}, $\RKer_{\gaStep}$ has a
unique invariant distribution $\pilgaStep$ and is $\VSIAM$-uniformly
geometrically ergodic: there exists $\kappaSIAM_{\convSym} \in \ooint{0,1}$ and $C_{\convSym}
\geq 0$ such that all $n \geq 0$ and $x \in \rset^d$, 
\begin{equation*}
  \tvnorm{\delta_x \RKer_{\gaStep}^n - \pilgaStep} \leq C_{\convSym} \VSIAM(x) \kappaSIAM_{\convSym}^n \eqsp.
\end{equation*}
Note $\pilgaStep$ is different from $\pil$, nevertheless the following
result shows that choosing $\gaStep$ small enough, the ULA generates
samples very close to the distribution $\pil$.
 
We are now ready to present our main theoretical result: a non-asymptotic bound of the total-variation distance between $\pi$ and the marginal laws of the samples generated by MYULA. Denote in the following by $\Fsmall : \rset_+ \to \rset_+$ the function given for all $r \geq 0$ by
\begin{equation}
  \label{eq:Fsmall}
  \Fsmall(r) = r^{2}/ \defEns{2\bfPhi^{-1}(3/4)}^{2} \eqsp.
\end{equation}
 \begin{theorem}[\protect{\cite[Corollary 19]{durmus:moulines:2015}}]
\label{theo:convergence_TV_dec-stepsize-convV}
Assume \Cref{assum:form-potential} and \Cref{assum:potentialUl}. Let $\bargaStep \in \ocint{0,\LUl^{-1}}$. For all $\varepsilon >0$ and $x \in \rset^d$,  we have 
$$
\tvnorm{\delta_x  \RKer_{\gaStep}^n-\pi} \leq \varepsilon \eqsp,
$$ 
provided that  $n > T \gaStep^{-1}$ with
\begin{align*}
T &= \max\defEns{32\, \rhoUl^{-2}\log\parenthese{8 \varepsilon^{-1}A_1(x)}, \log(16 \varepsilon^{-1}) \Big/(- \log(\kappa))} \\
 \gaStep &\leq \frac{-d+\sqrt{d^2 +(2/3) A_2(x) \varepsilon^2 (\LUl^2T)^{-1} }}{2 A_2(x)/3} \wedge \bargaStep  \eqsp,
\end{align*}
where $\KOneD = \max(1,4 d / \rhoUl,\RUl) $
\begin{align*}
 \bOne &=(\rhoUl/4)\parentheseDeux{\rhoUl \KOneD /4 +d}
\max \defEns{1,(\KOneD^2 +1)^{-1/2} \exp(\rhoUl( \KOneD^2+1)^{1/2}/4)} \\
A_1(x) &=  (1/2)( \VSIAM(x) + \bSIAM(-\lambdaSIAM^{\gaStep} \log(\lambdaSIAM))^{-1} + 8 \rhoUl^{-2} \bOne)+ 16 \rhoUl^{-2} \bOne\rme^{32^{-1} \rhoUl^{2} \Fsmall\defEns{(8/\rhoUl) \log( 32 \rhoUl^{-2} \bOne)}} \\
  A_2(x) & = L^2\parenthese{4 \rhoUl^{-1} \parentheseDeux{ 1+\log \defEns{\VSIAM(x) + \bSIAM(-\lambdaSIAM^{\gaStep} \log(\lambdaSIAM))^{-1}}}}^2\\
\log(\kappa) 
&= -  \log(2) (\rhoUl^{2}/32) \parentheseDeux{\log \defEns{8 \rhoUl^{-2} \bOne \parenthese{3+ 4 \rhoUl^{-2}\rme^{32^{-1} \rhoUl^{2} \Fsmall\defEns{(8/\rhoUl) \log( 32 \rhoUl^{-2} \bOne)}}}} +\log(2) }^{-1}   \eqsp, \\
\end{align*}
$ \KOne,\lambdaSIAM,\bSIAM,\VSIAM$ are defined in \eqref{eq:convex_drift} and $\Fsmall$ in \eqref{eq:Fsmall}.
\end{theorem}
\begin{proof}
The proof follows from combining  \cite[Lemma 4, Theorem 14, Theorem 16]{durmus:moulines:2015}.
\end{proof}
This result implies that the number of iteration to reach a precision target $\varepsilon$ is, at worse, of order $d^5\log^2(\varepsilon^{-1})\varepsilon^{-2}$ for this class of models. Significantly more precise bounds can be obtained under more stringent assumption on $\Ul$. In particular, we consider the case where $\Ul$ is strongly convex outside some ball; see \cite{eberle:2015}.
\begin{assumption}
  \label{assum:strongConvexityOutsideBallDriftV}
  There exist $\RSt \geq 1$ and $\constStV >0$, such that for all
  $x,y \in \rset^d$, $\norm{x-y} \geq \RStV$,
  \[
  \ps{\nabla \Ul(x) -\nabla \Ul(y)} {x-y} \geq  \constStV
  \norm[2]{x-y} \eqsp.
  \]
\end{assumption}
Of course, in the case where $\fU$ is strongly convex then this assumption holds.
\begin{theorem}[\protect{\cite[Lemma 4, Theorem 21]{durmus:moulines:2015}}]
\label{theo:convergence_TV_dec-stepsize-StV}
Assume  \Cref{assum:form-potential} and \Cref{assum:strongConvexityOutsideBallDriftV}. 
Let $\bargaStep \in \ocint{0,\LUl^{-1}}$. Then for all $\varepsilon >0$, we get $\tvnorm{\delta_x  \RKer_{\gaStep}^n-\pi} \leq \varepsilon$ provided that  $n > T \gaStep^{-1}$ with
\begin{align*}
T &= \parenthese{ \log\{ A_1(x) \}-\log(\varepsilon/2)} \Big/(- \log(\kappa))\\
 \gaStep &\leq \frac{-d+\sqrt{d^2 +(2/3) A_2(x) \varepsilon^2 (\LUl^2T)^{-1} }}{2 A_2(x)/3} \wedge \bargaStep  \eqsp,
\end{align*}
where
\begin{align*}
A_1(x) &= 5+ \parenthese{d/\constStV + \RStV^2}^{1/2} +(A_1(x)/L^2)^{1/2} \\
  A_2(x) & = L^2 \parenthese{\norm[2]{x-\xstar} + 2(d + \constStV  \RSt^2)(\rme^{-\gaStep(2\constSt+\bargaStep \LUl^2)}/(2\constSt+\bargaStep \LUl^2) )^{-1}}\\
 \log(\kappa) 
 &= -  (\log(2) \constSt/2) \parentheseDeux{\log \defEns{ \parenthese{1+ \rme^{\constSt \Fsmall\defEns{\max(1,\RStV)}/4}}\parenthese{1+\max(1,\RStV)}} +\log(2) }^{-1}  
 \eqsp,
\end{align*}
and $\Fsmall$ is given in \eqref{eq:Fsmall}.
\end{theorem}
This result implies that the worst minimal number of iterations to achieve a
precision level $\varepsilon >0$ is this time of order $d \log(d) \log^{2}(\varepsilon^{-1})\varepsilon^{-2}$.

\subsection{Selection of $\lambdaMY$ and $\gaStep$}\label{guidelines}
We now discuss practical guidelines for setting the values for $\lambdaMY$ and for $\gaStep$. As mentioned previously, our aim is to provide an efficient computation methodology that can be applied straightforwardly to any model satisfying \Cref{assum:form-potential}. Hence, rather than seeking optimal values for specific models, we focus on general rules that are simple, robust, and which only involve tractable quantities such as Lipschitz constants. 

First, by Theorem \ref{theo:convergence_TV_dec-stepsize-convV}, $\gaStep$ should take its value in the range $\gaStep \in (0, \lambda/(\LL\lambda + 1)]$ to guarantee the stability of the Euler-Maruyama discretisation, and where we recall that $\LL$ is the Lipschitz constant of $\nabla f$. The values of $\gaStep$ within this range are subject to the a bias-variance trade-off. Precisely, large values of $\gaStep$ produce a fast-moving chain that convergences quickly and has low estimation variance, but potentially relatively high asymptotic bias. Conversely, small values of $\gaStep$ lead to low asymptotic bias, but produce a Markov chain that moves slowly and requires a large number of iterations to produce a stable estimate (such chains often also suffer from some additional bias from the transient or burn-in period). Because applications in imaging sciences involve high dimensionality and require moderately low computing times, as a general rule we recommend setting $\gaStep$ to a relatively large value. For example, in our experiments we use 
$$
\gaStep \in \left[\lambda/5(\LL\lambda + 1), \lambda/2(\LL\lambda + 1)\right]\,.
$$
Observe that this range depends on the value of $\lambda$, which is also subject to a bias-variance tradeoff. Letting $\lambda \rightarrow 0$ to bring $\pi^\lambda$ close to $\pi$ reduces asymptotic bias, but forces $\gaStep \rightarrow 0$ and consequently reduces significantly the efficiency of the chain. Conversely, increasing the value of $\lambda$ accelerates the chain at the expense of some asymptotic bias. Based on our experience, and again with an emphasis on efficiency in high dimensional settings, we recommend using values of $\lambda$ in the order of $\LL^{-1}$ (there is no benefit in using larger values of $\lambda$ because $\gaStep$ saturates at $\LL^{-1}$). In all our experiments we use $\lambda = 1/\LL$ and $\gaStep \in [\LL^{-1}/10, \LL^{-1}/4]$ and obtain estimation errors of the order of $1\%$.

\subsection{Connections to the proximal Metropolis-adjusted Langevin algorithm}
We conclude this section with a discussion of the connections between the proposed MYULA method and the original proximal Metropolis-adjusted Langevin algorithm (Px-MALA) \cite{Pereyra2015}. That algorithm is also based on a Euler-Maruyama approximation of a Langevin SDE targeting a Moreau-Yoshide-type regularised approximation of $\pi$. However, unlike MYULA, that algorithm uses this approximation as proposal mechanism to drive a Metropolis-Hastings (MH) algorithm targeting $\pi$ (not the regularised approximation). The role of the MH is two-fold: it removes the asymptotic bias related to the approximations involved, and it provides a theoretical framework for Px-MALA by placing the scheme within the framework of MH algorithms (recall that many theoretical results regarding ULAs are very recent). However, as mentioned previously, the introduction of the MH step often slows down the algorithm, thus leading to higher estimation variance and longer chains (and potentially some bias from the chain's initial transient regime). Of course, it also introduces a significant computational overhead related to the computation of the MH acceptance ratio \cite{Pereyra2015}. Another importance difference between MYULA and Px-MALA is that the latter uses the proximal operator of $U$, which is often unavailable and has to be approximated by using a forward-backward scheme based on the decomposition $U = f + g$ that we also use in this paper. This approximation error is corrected in practice by the MH step, but it is not considered in the theoretical analysis of the algorithm. Conversely, in MYULA this decomposition is explicit, both in the computational aspects of the method as well as in its theoretical analysis. Furthermore, the theory for MYULA presented in this paper is significantly more complete than that currently available for Px-MALA and other MALAs. Finally, MYULA is also more robust and simple to implement than Px-MALA. For example, identifying suitable values of $\gaStep$ for MYULA is straightforward by using the guidelines described above, whereas setting $\gaStep$ for Px-MALA can be challenging and often requires using an adaptive MCMC approach based on a stochastic approximation scheme \cite{Pereyra2015,Green2015}.



\section{Experimental results}
\label{sec:experiments}
In this section we illustrate the proposed methodology with four canonical imaging inverse problems related to image deconvolution and tomographic reconstruction with total-variation and $\ell_1$ sparse priors. In the Bayesian setting these problems are typically solved by MAP estimation, which delivers accurate solutions and can be computed very efficiently by using proximal convex optimisation algorithm. Here we demonstrate MYULA by performing some advanced and challenging Bayesian analyses that are beyond the scope of optimisation-based mathematical imaging methodologies. For example, in Section \ref{exp:BMS} we report two experiments where we use MYULA to perform Bayesian model choice for image deconvolution models, and where a novelty is that comparisons are performed intrinsically (i.e., without ground truth available) by computing the posterior probability of each model given the observed data. Following on from this, in Section \ref{exp:BUQ} we report the two additional experiments where we use MYULA to explore the posterior uncertainty about $x$ and analyse specific aspects about the solutions delivered, particularly by computing simultaneous credible sets (joint Bayesian confidence sets).

Moreover, to assess the computational efficiency and the accuracy of MYULA we benchmark our estimations against the results of Px-MALA \cite{Pereyra2015} targeting the exact posterior $\pi(x) = p(x|y)$ (recall that this algorithm has no asymptotic estimation bias). We emphasise at this point that we do not seek to compare explicitly and quantitatively the methods because: 1) MYULA and Px-MALA do not target the exact same stationary distribution; 2) high-dimensional quantitative efficiency comparisons may depend strongly on the summary statistics used to define the efficiency metrics; and 3) results can often be marginally improved by fine tuning the algorithm parameters (e.g., step sizes, burn-in periods, etc.). What our comparisons seek to demonstrate is that MYULA can deliver reliable approximate inferences with a computational cost that is often significantly lower than Px-MALA, and more importantly, that it provides a general, robust, and theoretically sound computational framework for performing advanced Bayesian analyses for imaging problems. Experiments were conducted on a Apple Macbook Pro computer running MATLAB 2015. 

\subsection{Bayesian model selection}\label{exp:BMS}
\subsubsection{Bayesian analysis and computation}
Most mathematical imaging problems can be solved with a range of alternative models. Currently, the predominant approach to select the best model for a specific problem is to compare their estimations against ground truth. For example, given $K$ alternative Bayesian models $\mathcal{M}_1, \ldots, \mathcal{M}_K$, practitioners often benchmark models by artificially degrading a set of test images, computing the MAP estimator for each model and image, and then measuring estimation error with respect to the truth. The model with the best overall performance is then used in applications to analyse real data. Of course this approach to model selection has some limitations: 1) it relies strongly on test data that may not be representative of the unknown, and 2) conclusions can depend on the estimation error metrics used.

An advantage of formulating inverse problems within the Bayesian framework is that, in addition to strategies to perform point estimation, this formalism also provides theory to compare models objectively and intrinsically, and hence perform model selection in the absence of ground truth. Precisely, $K$ alternative Bayesian models are compared through their marginal posterior probabilities
\begin{equation}\label{margPost}
p(\mathcal{M}_j | y) = \frac{ p(y|\mathcal{M}_j) K^{-1}}{\sum_{k = 1}^K p(y | \mathcal{M}_k) K^{-1}},\quad j = \{1, \ldots, K\}\, ,
\end{equation}
where for objectiveness here we use an uniform prior on the auxiliary variable $j$ indexing the models, $p(y | \mathcal{M}_j)$ is the marginal likelihood
\begin{equation}\label{margLike}
p(y | \mathcal{M}_j) = \int p(x,y | \mathcal{M}_j) \textrm{d}x,\quad j = \{1, \ldots, K\}\, ,
\end{equation}
measuring model-fit-to-data and $p(y,x | \mathcal{M}_j)$ is the joint
probability density associated with $\mathcal{M}_j$ (see
\Cref{sec:selection_model_case-proper-imp_prior} for details regarding the case of
improper priors). Following Bayesian decision theory, to perform model
selection we simply chose the model with the highest posterior
probability (this is equivalent to performing MAP estimation on the
model index $j$):
\begin{equation*}
\mathcal{M}^* = \argmax_{j \in \{1,\ldots,K\}} p(\mathcal{M}_j | y).
\end{equation*}

From a computation viewpoint, performing Bayesian model selection for imaging problems is challenging because it requires evaluating the likelihoods $p(y | \mathcal{M}_j)$ up to a proportionality constant, or equivalently the Bayes factors $p(y | \mathcal{M}_j)/p(y | \mathcal{M}_i)$ for $i,j \in \{1,\cdots,K\}$ (see \Cref{sec:case-proper-imp_prior} for details regarding the case of improper priors). Here we perform this computation by Monte Carlo integration. Precisely, given $n$ samples $X^M_1,\ldots, X^M_n$ from $p(x |y, \mathcal{M}_j)$, we approximate the marginal likelihood of model $\mathcal{M}_j$ by using the truncated harmonic mean estimator \citep{Robert2009AIP}
 \begin{equation}\label{harmonicEstimator}
 p(y|\mathcal{M}_j) \approx \left(\sum_{k = 1}^n \frac{\indi{\Setharm^{\star}}{(X^M_k)}}{p(X^M_k,y | \mathcal{M}_j)}\right)^{-1} \vol(\Setharm^{\star}) \eqsp, \quad j = \{1,2,3\}
 \end{equation}
where for all $x,y$, $p(x,y| \mathcal{M}_j)$ is joint density of $\mathcal{M}_j$ and $\Setharm^{\star} = \cup_{j = 1}^3 \hpd_{j,\alpha}$ is the union of highest posterior density regions \eqref{HPD} of each model at level $(1-\alpha)$ (see Section \ref{exp:BUQ} for details about HPD regions). In our experiments we use the samples to calibrate each $\hpd_{j,\alpha}$ for $\alpha = 0.8$. Notice that it is not necessary to compute $\vol(\Setharm^{\star})$ to calculate \eqref{margPost} because the normalisation is retrieved via $ \sum_{j=1}^3 p(\mathcal{M}_j|y) =1$. See \Cref{HME} for more details about this estimator and its use to compute the Bayes factors.

\subsubsection{Experiment 1: Image deconvolution with total-variation prior}\label{ssec:exp1}
\paragraph{Experiment setup}
To illustrate the Bayesian model selection approach we consider an image deconvolution problem with three alternative models related to three different blur operators. The goal of image deconvolution is to recover a high-resolution image $x \in \mathbb{R}^n$ from a blurred and noisy observation $y = H x + w$, where $H$ is a circulant blurring matrix and $w \sim \mathcal{N}(0,\sigma^2\boldsymbol{I}_n)$. This inverse problem is ill-conditioned, a difficulty that Bayesian image deconvolution methods address by exploiting the prior knowledge available. For this first experiment we consider three alternative models involving three different blur operators $H_1$, $H_2$, and $H_3$. With regards to the prior, we use the popular total-variation prior that promotes regularity by using the pseudo-norm $TV(x) = \|\nabla_d x\|_{1-2}$, where $\|\cdot\|_{1-2}$ is the composite $\ell_1 -\ell_2$ norm and $\nabla_d$ is the two-dimensional discrete gradient operator. The posterior distribution $p(x|y)$ for the models is given by 
\begin{eqnarray}\label{deconvolutionTV}
\mathcal{M}_j : \quad \pi(x) \propto \exp{\left[-(\|y-H_j x\|^2/2\sigma^2) - \beta TV(x) \right]}
\end{eqnarray}
with fixed hyper-parameters $\sigma > 0$ and $\beta > 0$ set manually by an expert. This density is log-concave and MAP estimation can be performed efficiently by proximal convex optimisation (here we use the ADMM algorithm SALSA \citep{Figueiredo2011}).

Figure \ref{FibBoat1} presents an experiment with the \texttt{Boat} test image of size $d = 256 \times 256$ pixels. Figure \ref{FibBoat1}(a) shows a blurred and noisy observation $y$, generated by using a $5 \times 5$ uniform blur and Gaussian noise with $\sigma = 0.47$, related to a blurred signal-to-noise ratio of $40$dB. Moreover, Figures \ref{FibBoat1}(b)-(d) show the MAP estimates associated with three alternative instances of model \eqref{deconvolutionTV} involving the following blur operators:
\begin{itemize}
\item $\mathcal{M}_1$: $H_1$ is the correct $5 \times 5$ uniform blur operator.
\item $\mathcal{M}_2$: $H_2$ is a mildly misspecified $6 \times 6$ uniform blur operator. 
\item $\mathcal{M}_3$: $H_3$ is a strongly misspecified $7 \times 7$ uniform blur operator.
\end{itemize}
(All models share the same hyper-parameter values $\sigma = 0.47$ and $\beta = 0.03$ selected manually to produce good image deconvolution results.)  We observe in Figure \ref{FibBoat1} that models $\mathcal{M}_1$ and $\mathcal{M}_2$ have produced sharp images with fine detail, whereas $\mathcal{M}_3$ is clearly misspecified. In terms of estimation performance with respect to the truth, as expected the estimate of Figure \ref{FibBoat1}(c) corresponding to model $\mathcal{M}_1$ achieves the highest peak signal-to-noise-ratio (PSNR) of $33.8$dB,  $\mathcal{M}_2$ scores $33.4$dB, and $\mathcal{M}_3$ scores $13.4$dB. Finally, computing the MAP estimates displayed in Figure \ref{FibBoat1} with SALSA \cite{Figueiredo2011} required $2$ seconds per model.

\begin{figure}
\begin{minipage}[l2]{0.49\linewidth}
  \centering
  \centerline{\includegraphics[width=7.5cm]{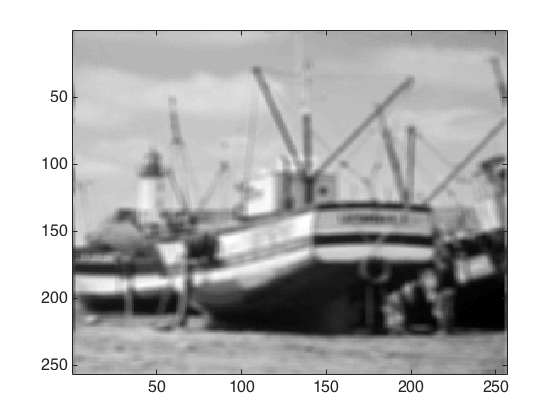}}
  \small{(a)}
\end{minipage}
\begin{minipage}[l2]{0.49\linewidth}
  \centering
  \centerline{\includegraphics[width=7.5cm]{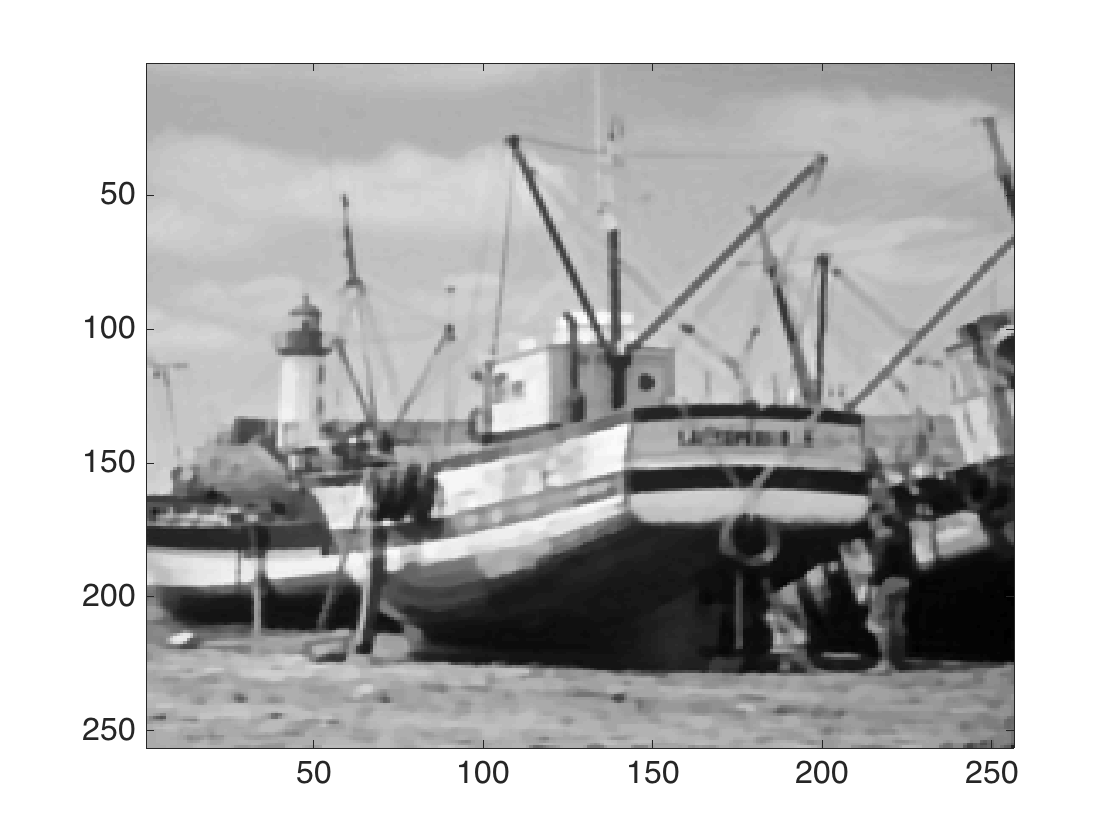}}
  \small{(b)}
\end{minipage}
\begin{minipage}[l2]{0.49\linewidth}
  \centering
  \centerline{\includegraphics[width=7.5cm]{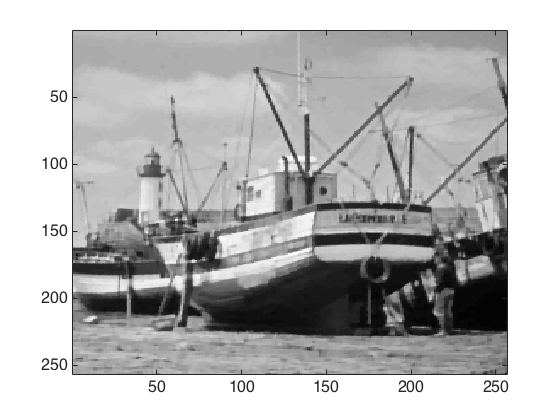}}
  \small{(c)}
\end{minipage}
\begin{minipage}[l2]{0.49\linewidth}
  \centering
  \centerline{\includegraphics[width=7.5cm]{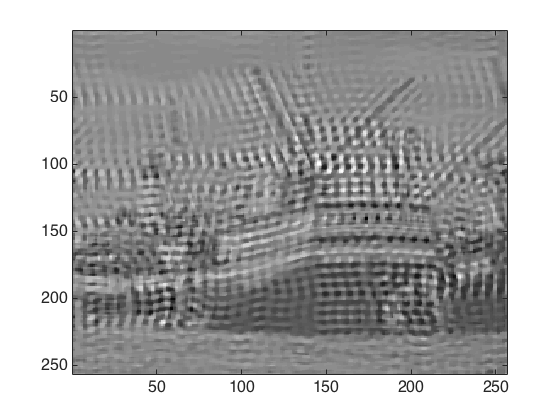}}
  \small{(d)}
\end{minipage}
\caption{\small{Deconvolution experiment - \texttt{Boat} test image ($256\times 256$ pixels): (a) Blurred and noisy image $y$, (b)-(d) MAP estimators corresponding to models $\mathcal{M}_1$, $\mathcal{M}_2$, and $\mathcal{M}_3$.}} \label{FibBoat1}
\end{figure}

\paragraph{Model selection in the absence of ground truth}
We now demonstrate the Bayesian approach to perform model selection intrinsically. Precisely, we ran $10^5$ iterations of  MYULA  with the specific blur operators corresponding to $\mathcal{M}_1$, $\mathcal{M}_2$, and $\mathcal{M}_3$. For this experiment we implemented MYULA with $f(x) = \|y-H_j x\|^2/2\sigma^2$ and $g(x) = \beta TV(x)$, with fixed algorithm parameters $\lambda = L_f^{-1} = 0.45$ and $\gamma = L_f^{-1}/5 = 0.1$, and by using Chambolle's algorithm \citep{Chambolle} to evaluate the proximal operator of the TV-norm. Computing these samples required approximately $30$ minutes per model. Following on from this, we used the samples to calibrate the high-posterior-density regions $\hpd_j$ of each model at level $20\%$, and then computed the Bayes factors between the models by using \eqref{harmonicEstimator} (see \ref{sec:case-proper-prior} for details).

By applying this procedure we obtained that $\mathcal{M}_1$ has the highest posterior probability $p(\mathcal{M}_1|y) = 0.964$, followed by $p(\mathcal{M}_2|y) = 0.036$ and $p(\mathcal{M}_3|y) < 0.001$ (the values of the Bayes factors for this experiment are $\hat{B}_{1,2}(y) = 26.8$ and $\hat{B}_{1,3}(y) > 10^{3}$). These results, which have been computing without using any form of ground truth, are in agreement with the PSNR values calculated by using the true image and provide strong evidence in favour of model $\mathcal{M}_1$. They also confirm the good performance of the Bayesian model selection technique.

\paragraph{Comparison with proximal MALA}
We conclude this first experiment by benchmarking our estimations against Px-MALA, which targets \eqref{deconvolutionTV} exactly. Precisely, we recalculated the models' posterior probabilities \eqref{margPost} with Px-MALA and obtained that $p( \mathcal{M}_1 | y) = 0.962$, $p(  \mathcal{M}_2 | y) = 0.038$, and $p( \mathcal{M}_3 | y) < 0.001$, indicating that the MYULA estimate has an approximation error of the order of $0.5\%$ (to obtain accurate estimates for Px-MALA we used $n = 10^7$ iterations with an adaptive time-step targeting an average acceptance rate of order $45\%$). Moreover, comparing the chains generated with MYULA and Px-MALA revealed that MYULA is significantly more computationally efficient than Px-MALA. For illustration, Fig. \ref{FibBoat2}(a) shows the transient regimes of the MYULA and Px-MALA chains related $\mathcal{M}_1$, where starting from a common initial condition the chains converge to the posterior typical set\footnote{In stationarity, $x|y$ is with very high probability in the neighbourhood of the $(d-1)$-dimensional shell $\{x : U(x) = \mathbb{E}[U(x)|y]\}$, see \cite{Pereyra:2016b}} of $p(x|y)$ (to improve visibility this is displayed in logarithmic scale). Observe that MYULA requires around $10^2$ iterations to navigate the parameter space and reach the typical set, whereas Px-MALA requires $10^4$ iterations. Furthermore, to compare the efficiency of the chains in stationarity, Fig. \ref{FibBoat2}(b) shows the autocorrelation function of the chains generated by MYULA and Px-MALA. To highlight the efficiency of MYULA we have used the chains' slowest component (i.e., that with largest variance) as summary statistic. Again, observe that MYULA is clearly significantly more efficient than Px-MALA. From a practitioner's viewpoint, this efficiency advantage is further accentuated by the fact that MYULA iterations are almost twice less computationally expensive than Px-MALA iterations, which include the MH step.

\begin{figure}
\begin{minipage}[l2]{0.5\linewidth}
  \centering
 \centerline{\includegraphics[width=7.5cm]{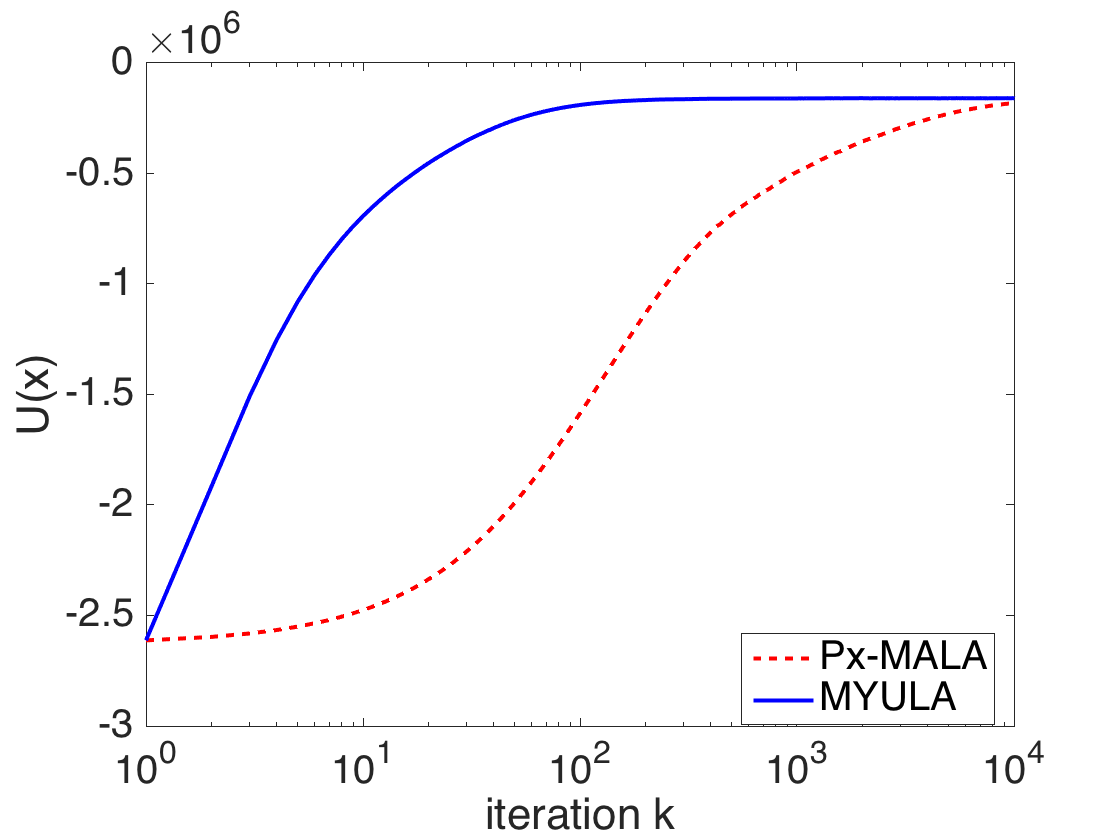}}
  \small{(a)}
\end{minipage}
\begin{minipage}[l2]{0.5\linewidth}
  \centering
  \centerline{\includegraphics[width=7.5cm]{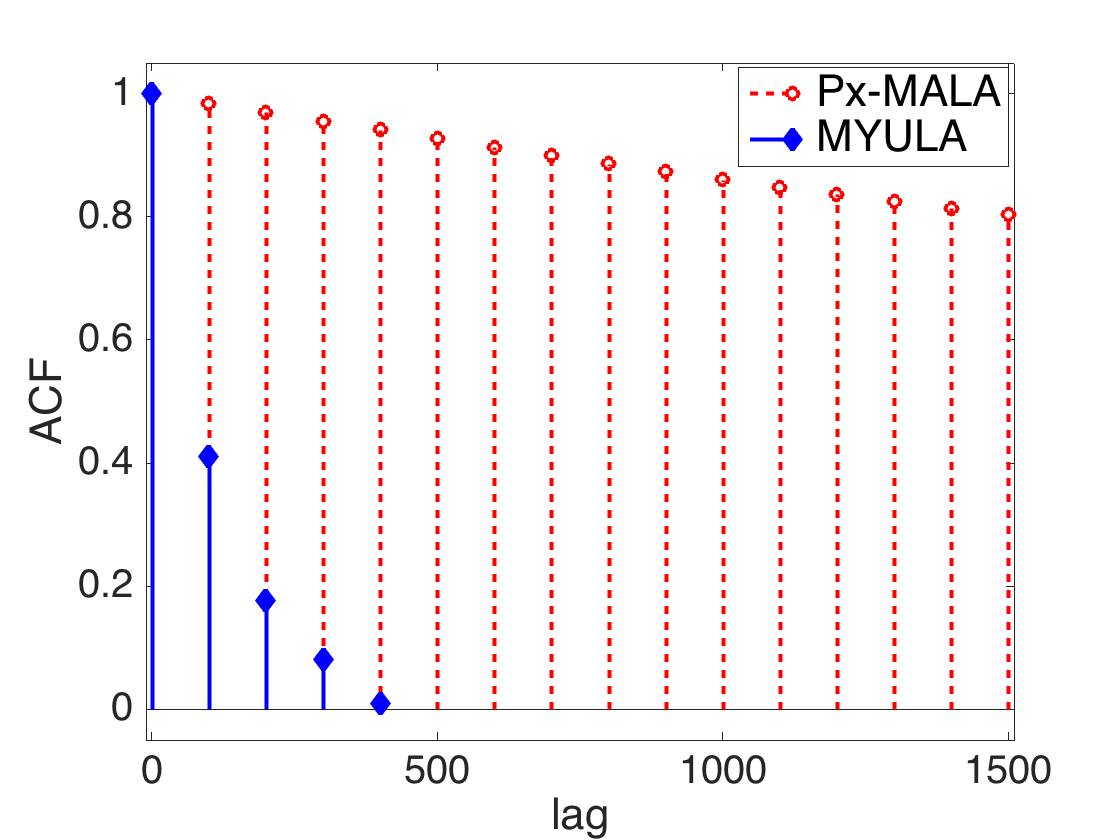}}
  \small{(b)}
\end{minipage}
\caption{\small{MYULA and Px-MALA comparison: (a) Convergence of the chains to the typical set of \eqref{deconvolutionTV} under model $\mathcal{M}_1$ (logarithmic scale), (b) chain autocorrelation function (ACF).}} \label{FibBoat2}
\end{figure}

\subsubsection{Experiment 2: Image deconvolution with wavelet frame}
\label{sec:experiment-2:-image}
\paragraph{Experiment setup}
The second model selection experiment we consider involves three alternative image deconvolution models with different priors. This experiment is more challenging than the previous one because priors operate indirectly on $y$ through $x$. We consider three models of the form
\begin{eqnarray}\label{deconvolutionL1wave}
\mathcal{M}_j: \quad p(x|y) \propto \exp{\left[-(\|y-H x\|^2/2\sigma^2) - \beta_j \|\Psi_j x\|_1 \right]}
\end{eqnarray}
where $\Psi_j$ is a model dependent frame:
\begin{itemize}
\item $\mathcal{M}_1$: $\Psi_1$ is a redundant Haar frame with 6-level, and $\beta_1 = 0.02$ is selected automatically by using a hierarchical Bayesian method \citep{Pereyra_EUSIPCO_2015},
\item $\mathcal{M}_2$: $\Psi_2$ is a redundant Haar frame with 3-level, and $\beta_2 = 0.02$ is selected automatically by using a hierarchical Bayesian method \citep{Pereyra_EUSIPCO_2015},
\item $\mathcal{M}_3$: $\Psi_3$ is a redundant Haar frame with 3-level, and $\beta_3 = 0.003$ is selected automatically by using the L-curve method \citep{Hanke1993}.
\end{itemize}
To make the selection problem even more challenging, in this experiment we use a higher noise level $\sigma = 1.76$, related to a blurred signal-to-noise ratio of $30$dB. We note that \eqref{deconvolutionL1wave} is log-concave and MAP estimation can be performed efficiently by proximal convex optimisation (here we use the ADMM algorithm SALSA \cite{Figueiredo2011}).

Fig. \ref{FibFlin} presents an experiment with the \texttt{Flinstones} test image of size $d = 256 \times 256$ pixels. Fig. \ref{FibBoat1}(a) shows the blurred and noisy observation $y$ used in this experiment, which we generated by using a $5 \times 5$ uniform blur and $\sigma = 1.76$, and Fig. \ref{FibFlin}(b)-(d) show the MAP estimates obtained with $\mathcal{M}_1$, $\mathcal{M}_2$, and $\mathcal{M}_3$ by using SALSA \citep{Figueiredo2011} (these computations required $4$ seconds per model). We observe in Figure \ref{FibBoat1} that models $\mathcal{M}_1$ and $\mathcal{M}_2$ have produced sharp images with fine detail, whereas $\mathcal{M}_3$ is misspecified. In terms of estimation performance with respect to the truth, the estimate of Figure \ref{FibFlin}(c) corresponding to model $\mathcal{M}_2$ achieves the highest peak signal-to-noise-ratio (PSNR) of $20.8$dB,  $\mathcal{M}_1$ scores $20.6$dB, and $\mathcal{M}_3$ scores $11.6$dB.

\begin{figure}
\begin{minipage}[l2]{0.49\linewidth}
  \centering
  \centerline{\includegraphics[width=7.5cm]{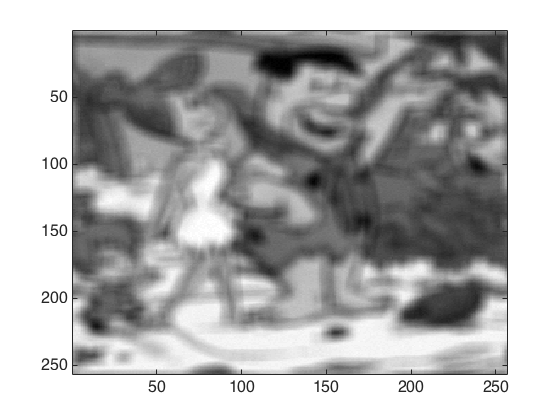}}
  \small{(a)}
\end{minipage}
\begin{minipage}[l2]{0.49\linewidth}
  \centering
  \centerline{\includegraphics[width=7.5cm]{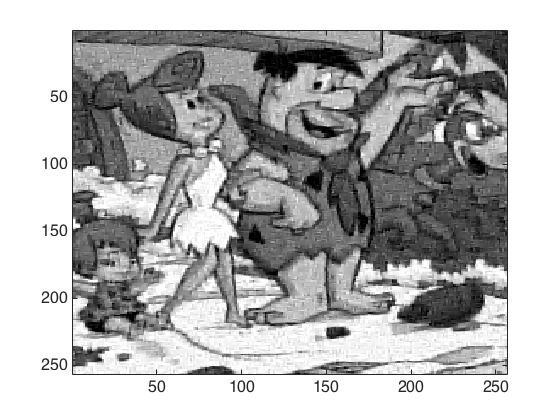}}
  \small{(b)}
\end{minipage}
\begin{minipage}[l2]{0.49\linewidth}
  \centering
  \centerline{\includegraphics[width=7.5cm]{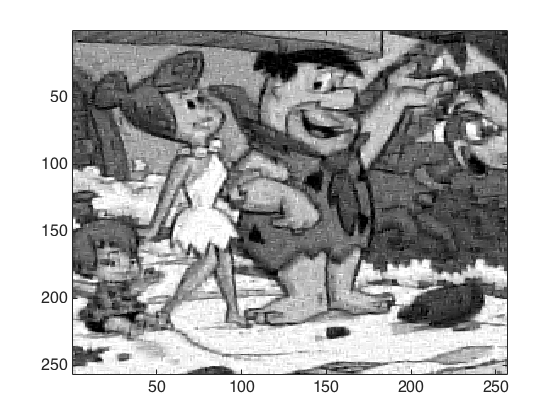}}
  \small{(c)}
\end{minipage}
\begin{minipage}[l2]{0.49\linewidth}
  \centering
  \centerline{\includegraphics[width=7.5cm]{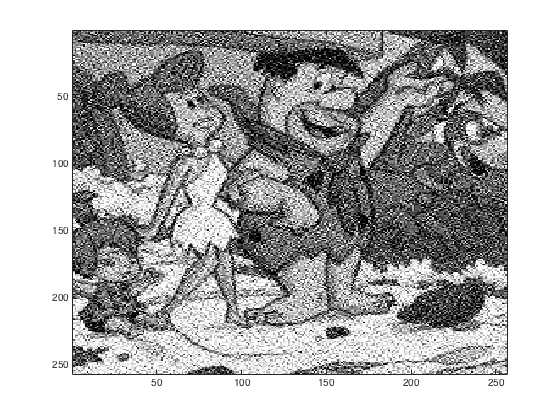}}
  \small{(d)}
\end{minipage}
\caption{\small{Deconvolution experiment - \texttt{Flinstones} test image ($256\times 256$ pixels): (a) Blurred and noisy image $y$, (b)-(d) MAP estimators corresponding to models $\mathcal{M}_1$, $\mathcal{M}_2$, and $\mathcal{M}_3$.}} \label{FibFlin}
\end{figure}

\paragraph{Model selection in the absence of ground truth}
Similarly to the previous experiment, we used MYULA to perform Bayesian model selection intrinsically. Precisely, we used MYULA to generate three sets of $n = 10^5$ samples $X^M_1,\ldots, X^M_n$ approximately distributed according to \eqref{deconvolutionL1wave} with the parameters corresponding to $\mathcal{M}_1$, $\mathcal{M}_2$, and $\mathcal{M}_3$. For this experiment we implemented MYULA with $f(x) = \|y-H x\|^2/2\sigma^2$ and $g(x) = \beta_j\|\Psi_j x\|_1$, with fixed algorithm parameters $\lambda = L_f^{-1} = 4.5$  and $\gamma = L_f^{-1}/5 = 0.9$. Computing these samples required $50$ minutes per model. Following on from this, we used the samples to calibrate the high-posterior-density regions $\hpd_j$ of each model at level $20\%$, and then computed the Bayes factors between the models by using \eqref{harmonicEstimator} (see \ref{sec:case-proper-prior} for details).

By applying this procedure we obtained that $\mathcal{M}_2$ has the highest posterior probability $p(\mathcal{M}_2|y) = 0.42$, followed by $p(\mathcal{M}_1|y) = 0.32$ and $p(\mathcal{M}_3|y) = 0.26$ (the values of the Bayes factors for this experiment are $\hat{B}_{2,1}(y) = 1.31$ and $\hat{B}_{2,3}(y) = 1.62$). Note that these results, which have been computing without using any form of ground truth, are in agreement with the PSNR values calculated by using the true image and indicate that $\mathcal{M}_2$ is the most appropriate model for data $y$.

\paragraph{Comparison with proximal MALA}
Again, we conclude our second experiment by benchmarking our estimations against Px-MALA, which targets \eqref{deconvolutionL1wave} exactly. Precisely, we recalculated the models' posterior probabilities \eqref{margPost} with Px-MALA and obtained that $p(y | \mathcal{M}_1) = 0.41$, $p(y | \mathcal{M}_2) = 0.33$, and $p(y | \mathcal{M}_3)= 0.26$, indicating that the MYULA estimate has an approximation error of the order of $0.5\%$ (to obtain accurate estimates for Px-MALA we used $n = 10^7$ iterations with an adaptive time-step targeting an average acceptance rate of order $45\%$). Moreover, efficiency analyses indicate that in this case MYULA is approximately an order of magnitude more efficient per iteration than Px-MALA, with an additional advantage in terms of time-normalised computational efficiency because of a lower computational cost per iteration. 

\subsection{Bayesian uncertainty quantification via posterior credible sets}\label{exp:BUQ}
\subsubsection{Bayesian analysis and computation}
As mentioned earlier, point estimators such as $\hat{x}_{MAP}$ deliver accurate results but do not provide information about the posterior uncertainty of $x$. Given the uncertainty that is inherent to ill-posed and ill-conditioned inverse problems, it would be highly desirable to complement point estimators with posterior credibility sets that indicate the region of the parameter space where most of the posterior probability mass of $x$ lies. This is formalised in the Bayesian decision theory framework by computing \emph{credible regions} \cite{cprbayes}. A set $C_\alpha$ is a posterior credible region with confidence level $(1-\alpha)$ if
\begin{equation*}
\mathbb{P} \left[x \in \hpdz_{\alpha} | y \right] = 1-\alpha.
\end{equation*}
It is easy to check that for any $\alpha \in (0,1)$ there are infinitely many regions of the parameter space that verify this property. Among all possible regions, the so-called \emph{highest posterior density} (HPD) region has minimum volume \cite{cprbayes}, and is given by 
\begin{eqnarray}\label{HPD}
\hpd_{\alpha} = \{ \bx : U(x) \leq \eta_\alpha \}
\end{eqnarray}
with $\eta_\alpha \in \mathbb{R}$ chosen such that $\int_{\hpd_{\alpha}} p(x|y) \textrm{d}\bx = 1-\alpha$ holds. This joint credible set has the important advantage that it can be enumerated by simply specifying the scalar value $\eta_\alpha$.

From a computation viewpoint, calculating credible sets for images is very challenging because it requires solving very high-dimensional integrals of the form $\int_{\hpd_{\alpha}} p(x|y)\textrm{d}x$. In this work, we use MYULA to approximate these integrals.

\subsubsection{Experiment 3: Tomographic image reconstruction}\label{tomographic_imaging}
\paragraph{Experiment setup}
The third experiment we consider is a tomographic image reconstruction problem with a total-variation prior. The goal is to recover the image $x \in \mathbb{R}^n$ from an incomplete and noisy set of Fourier measurements $y = A F x + w$, where $F$ is the discrete Fourier transform operator, $A$ is a tomographic sampling mask, and $w \sim \mathcal{N}(0,\sigma^2\boldsymbol{I}_n)$. This inverse problem is ill-posed, resulting in significant uncertainty about the true value of $x$. Similarly to Experiment 1, in this experiment we regularise the problem and reduce the uncertainty about $x$ by using a total-variation prior promoting piecewise regular images. The resulting posterior $p(x|y)$ is
\begin{eqnarray}\label{tomographic}
\pi(x) \propto \exp{\left[-\|y-A F x\|^2/2\sigma^2 -\beta TV(x)\right]}.
\end{eqnarray}
with fixed hyper-parameters $\sigma > 0$ and $\beta > 0$ set manually by an expert. We note that this density is log-concave and MAP estimation can be performed efficiently by proximal convex optimisation (here we use the ADMM algorithm SALSA \citep{Figueiredo2011}).

Figure \ref{FigMRI1} presents an experiment with the \texttt{Shepp-Logan phantom} magnetic resonance image (MRI) of size $d = 128 \times 128$ pixels presented in Figure \ref{FigMRI1}(a). Figure \ref{FigMRI1}(b) shows a noisy tomographic measurement $y$ of this image, contaminated with Gaussian noise with $\sigma = 7 \times 10^{-2}$ (to improve visibility Figure \ref{FigMRI1}(b) shows the amplitude of the Fourier coefficients in logarithmic scale, with black regions representing unobserved coefficients). Notice from Figure \ref{FigMRI1}(b) that only $15\%$ of the original Fourier coefficients are observed. Moreover, Figure \ref{FigMRI1}(c) shows the Bayesian estimate $\hat{x}_{MAP}$ associated with \eqref{tomographic} with hyper-parameter value $\beta = 5$. 

\begin{figure}[htbp!]
\begin{minipage}[l2]{0.49\linewidth}
  \centering
  \centerline{\includegraphics[width=7.5cm]{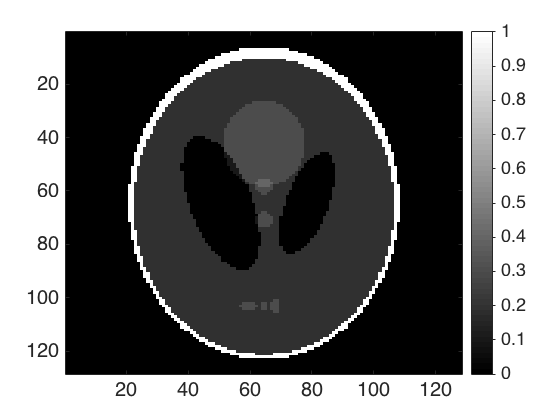}}
  \small{(a)}
\end{minipage}
\begin{minipage}[l2]{0.49\linewidth}
  \centering
  \centerline{\includegraphics[width=7.5cm]{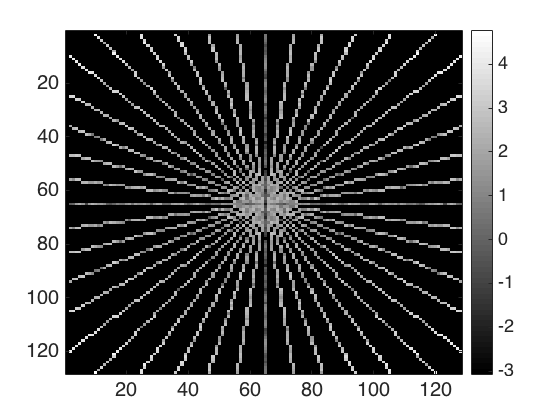}}
  \small{(b)}
\end{minipage}
\begin{minipage}[l2]{0.49\linewidth}
  \centering
  \centerline{\includegraphics[width=7.5cm]{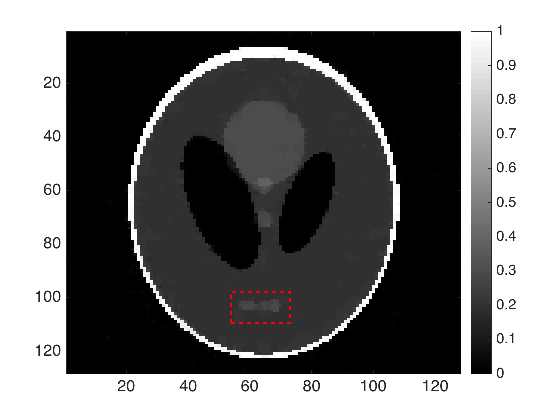}}
  \small{(d)}
\end{minipage}
\begin{minipage}[l2]{0.49\linewidth}
  \centering
  \centerline{\includegraphics[width=7.5cm]{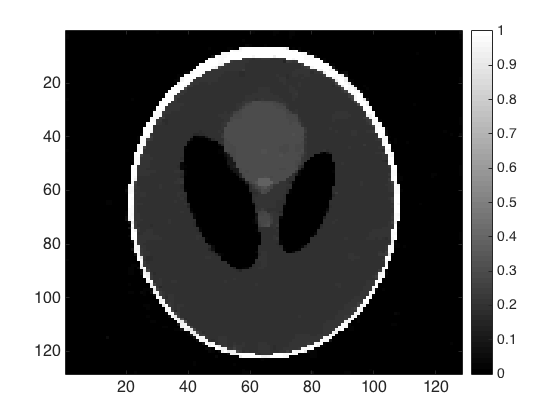}}
  \small{(d)}
\end{minipage}
\caption{\small{Tomography experiment: (a) \texttt{Shepp-Logan phantom} image ($128\times128$ pixels), (b) tomographic observation $y$ (amplitude of Fourier coefficients in logarithmic scale), (c) MAP estimator.}} \label{FigMRI1}
\end{figure}

\paragraph{Bayesian uncertainty analysis}
We now conduct a simple Bayesian uncertainty analysis to illustrate how posterior credible sets can inform decision-making. For illustration, suppose that the structure highlighted in red in Figure \ref{FigMRI1}(c) is relevant from a clinical viewpoint because it provides important information for diagnosis or treatment related decision-making. Also, suppose that we first observe this structure in the Bayesian estimate $\hat{x}_{MAP}$ and that, following on from this, we wish to explore the posterior uncertainty about $x$ to learn more about the structure. In particular, here we conduct a simple analysis to show that there is lack of confidence regarding the presence of this structure in the true image (i.e., the structure could be an artefact). Precisely, this is achieved by computing the HDP credible region $\hpd_{\alpha}$ and showing that it includes solutions that are essentially equivalent to $\hat{x}_{MAP}$ except for the fact that they do not have the structure of interest.

As alternative solution or ``counter example'' of $\hat{x}_{MAP}$, consider the image $x_\dagger$ displayed in Figure \ref{FigMRI1}(d). This image is equivalent to $\hat{x}_{MAP}$ except for the fact that the structure of interest has been removed (we generated this image by modifying $\hat{x}_{MAP}$ by applying a segmentation-inpainting process to replace the structure with the surrounding intensity level). Of course, clinicians observing $x_\dagger$ images would potentially arrive to significantly different conclusions about the diagnosis or the treatment required. This test image scores $U(x_\dagger) = 1.27 \times 10^4$.

To determine if $x_\dagger$ belongs to $\hpd_{\alpha}$ we used MYULA to generate $n = 10^5$ samples from \eqref{tomographic}, and calculated the HPD threshold $\eta_\alpha$ by estimating the $(1-\alpha)$-quantile of $U(x)$ (we implemented the algorithm with $f(x) = \|y-A F x\|^2/2\sigma^2$ and $g(x) = \beta TV(x)$, with fixed parameters $\lambda = L_f^{-1} = 1 \times 10^{-4}$ and $\gamma_k = L_f^{-1}/10 = 10^{-5}$, and by using Chambolle's algorithm \citep{Chambolle} to evaluate the proximal operator of the TV-norm). Fig. \ref{FigMRI3}(a) shows the threshold values $\eta_\alpha$ for a range of values of $\alpha \in [0.01,0.99]$. Observe that $U(x_\dagger) = 1.27 \times 10^4$ is significantly lower than the values displayed in Fig. \ref{FigMRI3}(a), indicating that the counter example image $x_\dagger$ belongs to set of likely solutions to the inverse problem (e.g., at level $90\%$ $\eta_{0.10} = 2.34 \times 10^4$ hence $x_\dagger \in \hpd_{0.10}$). Based on this we conclude that, with the current number of observations and noise level, it is not possible to assert confidently that the structure considered is present in the true image. Consequently, we would recommend that this data is not used as primary evidence to support decision-making about this structure. Generating the Monte Carlo samples and computing the HPD threshold values required $15$ minutes.  

\paragraph{Comparison with proximal MALA}
We conclude this experiment by benchmarking our estimations against Px-MALA, which targets \eqref{tomographic} exactly (to obtain accurate estimates for Px-MALA we use $n = 10^7$ iterations with an adaptive time-step targeting an average acceptance rate of order $45\%$). The HPD threshold values $\eta_\alpha$ obtained with Px-MALA are reported in Fig. \ref{FigMRI3}(a), notice the approximation error of order of $3\%$ due to MYULA's estimation bias (this does not affect the conclusions of the experiment). With regards to computational performance, an efficiency analysis of the two algorithms indicates that for this model MYULA is approximately two orders of magnitude more efficient than Px-MALA in terms of integrated autocorrelation time (for illustration Fig. \ref{FigMRI3}(b) compares the autocorrelation functions for slowest component of the MYULA and Px-MALA chains).

\begin{figure}[htbp!]
\begin{minipage}[l2]{0.49\linewidth}
  \centering
  \centerline{\includegraphics[width=7.5cm]{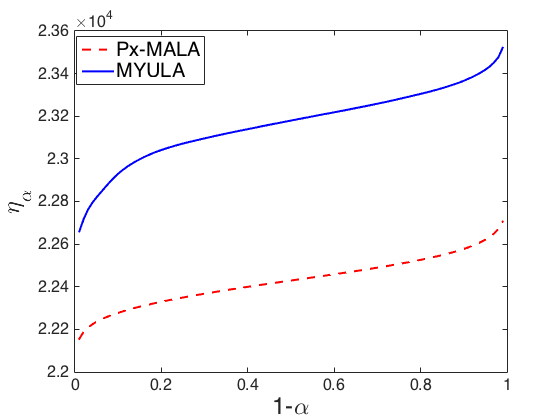}}
  \small{(a)}
\end{minipage}
 \begin{minipage}[l2]{0.49\linewidth}
   \centering
   \centerline{\includegraphics[width=7.5cm]{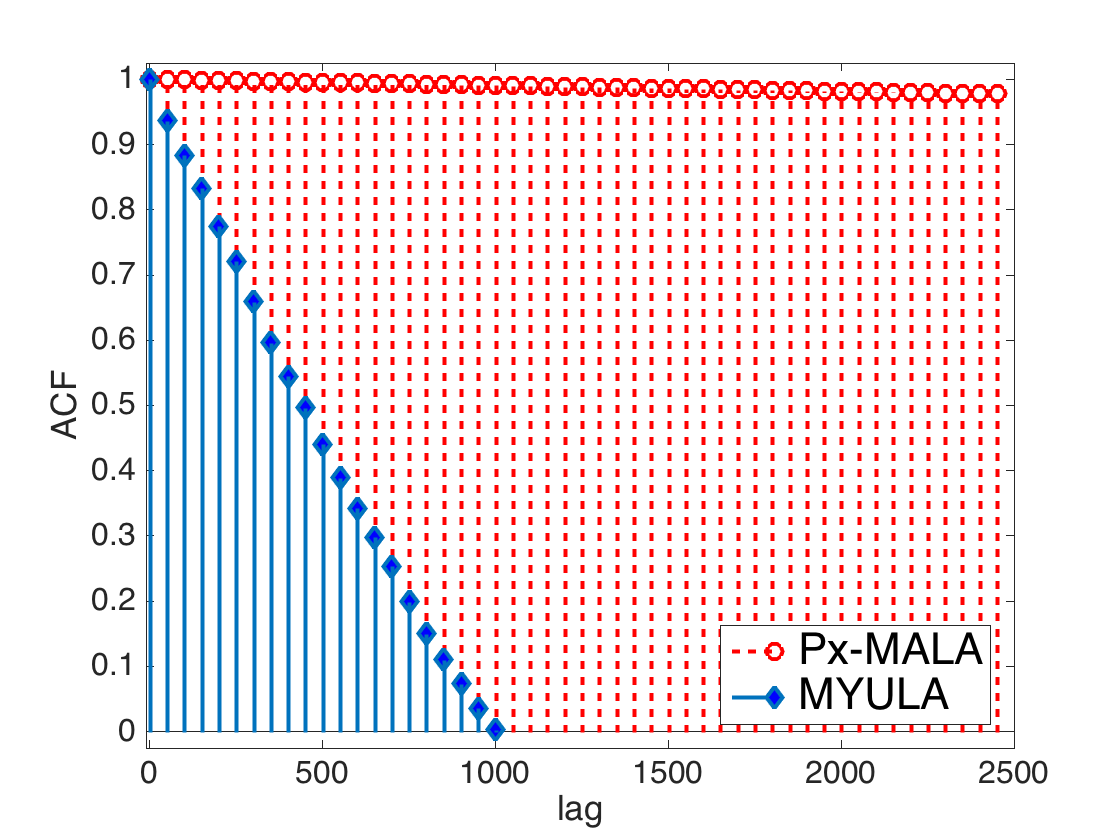}}
   \small{(b)}
 \end{minipage}

\caption{\small{Tomography experiment: (a) HDP region thresholds $\eta_\alpha$ for MYULA and Px-MALA, (b) chain autocorrelation functions for MYULA and Px-MALA.}} \label{FigMRI3}
\end{figure}

\subsubsection{Experiment 4: Sparse image deconvolution with an $\ell_1$ prior}\label{ssec:exp4}
\paragraph{Experiment setup}
The fourth experiment we consider is a sparse image deconvolution problem with a Laplace or $\ell_1$ prior. Again, we aim to recover $x \in \mathbb{R}^n$ from $y = Hx + \bw$, where $H$ is a circulant blurring matrix and $\bw \sim \mathcal{N}(0,\sigma^2\boldsymbol{I}_n)$. We expect sparse solutions and use a Laplace prior related to the $\ell_1$ norm of $x$. The resulting posterior $p(x|y)$ is
\begin{eqnarray}\label{deconvolution}
\pi(x) \propto \exp{\left[-\|y-Hx\|^2/2\sigma^2 - \beta \|x\|_{1}\right]}.
\end{eqnarray}
with fixed hyper-parameters $\sigma >0$ and $\beta > 0$ set manually by an expert. Similarly to the previous experiments, we notice that this density is log-concave and MAP estimation can be performed efficiently by proximal convex optimisation.

Figure \ref{FigMicro} presents an experiment with a microscopy dataset of \cite{Zhu2012} related to high-resolution live cell imaging. Figure \ref{FigMicro}(a) shows an observation $y$ of field of size $4 \mu m \times 4 \mu m$ containing $100$ molecules. This low-resolution observation has been acquired with an instrument specific point-spread-function of size $16 \times 16$ pixels and a blurred signal-to-noise ratio of $20$dB (see \cite{Zhu2012} for more details). Figure \ref{FigMicro}(b) shows the Bayesian estimate $\hat{x}_{MAP}$ associated with \eqref{deconvolution} with hyper-parameter value $\alpha = 0.01$ (notice that $\hat{x}_{MAP}$ is displayed in logarithmic scale to improve visibility). Computing this estimate with SALSA \cite{Figueiredo2011} required $2.3$ seconds. 

\paragraph{Bayesian uncertainty analysis}
As second example of Bayesian uncertainty quantification, we use $\hpd_{\alpha}$ to examine the uncertainty about the position of the group of molecules highlighted in red in Fig. \ref{FigMicro}, which we assume to be relevant for an application considered. Precisely, we used $n = 10^5$ samples generated with MYULA to compute $\hpd_{\alpha}$ with $\alpha = 0.01$ related to the $99\%$ confidence level, and obtained the threshold value $\eta_{0.01} = 9.69 \times 10^4$. Following on from this, to explore $\hpd_{0.01}$ to quantify the uncertainty about the exact position of the molecules, we generated several surrogate test images by modifying $\hat{x}_{MAP}$ by displacing the molecules in different directions until these surrogates exit $\hpd_{0.01}$. Figure \ref{FigMicro}(c) shows the posterior uncertainty of the molecule positions (note that for visibility the figure focuses on the region of interest). This analysis reveals that the uncertainty at level $99\%$ is of the order of $\pm 5$ pixels vertically and $\pm 8$ pixels horizontally, corresponding to $\pm 78nm$ and $\pm 125nm$. It is worth mentioning that these results are in close in agreement with the experimental precision results reported in \cite{Zhu2012}, which identified an average precision of the order of $80nm$ for the one hundred molecules.

\begin{figure}[htbp!]
\begin{minipage}[l2]{0.49\linewidth}
  \centering
  \centerline{\includegraphics[width=7.5cm]{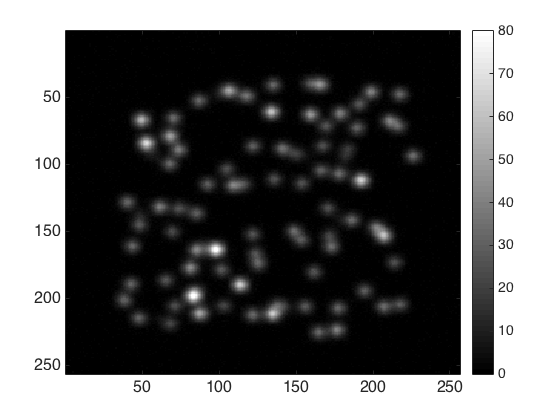}}
  \small{(a)}
\end{minipage}
\begin{minipage}[l2]{0.49\linewidth}
  \centering
  \centerline{\includegraphics[width=7.5cm]{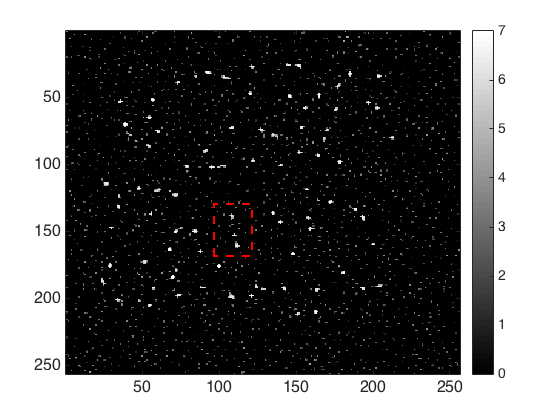}}
  \small{(b)}
\end{minipage}
\begin{minipage}[l2]{0.49\linewidth}
  \centering
  \centerline{\includegraphics[width=7.5cm]{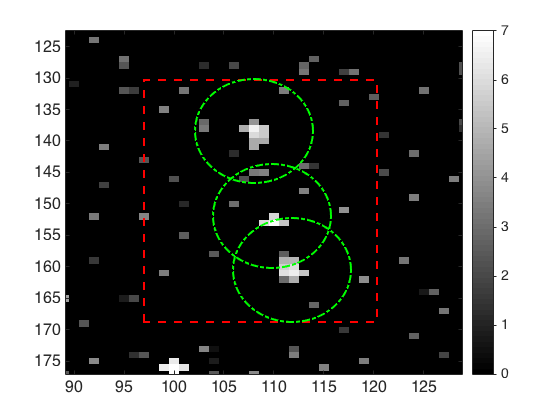}}
  \small{(c)}
\end{minipage}
\begin{minipage}[l2]{0.49\linewidth}
  \centering
  \centerline{\includegraphics[width=7.5cm]{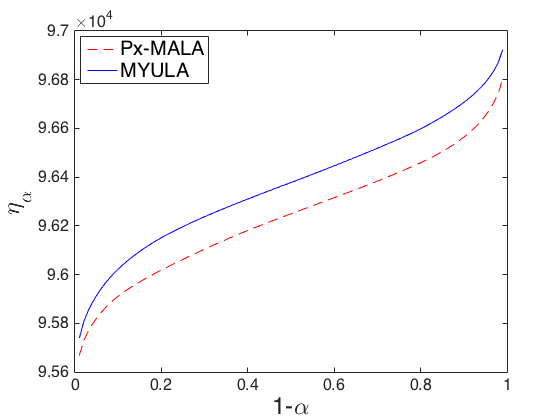}}
  \small{(d)}
\end{minipage}
\caption{\small{Microscopy experiment: (a) Blurred image $y$ ($256\times256$ pixels, $4\mu m \times 4\mu m$)),\newline (b) MAP estimate $\hat{x}_{MAP}$ (logarithmic scale), (c) molecule position uncertainty quantification (vertical: $\pm 78nm$, horizontal $\pm 125nm$), (d) HDP region thresholds $\eta_\alpha$ for MYULA and Px-MALA.}} \label{FigMicro}
\end{figure}

\paragraph{Comparison with proximal MALA}
Again, we conclude the experiment by benchmarking our estimations against Px-MALA, which targets \eqref{deconvolution} exactly (to obtain accurate estimates for Px-MALA we use $n = 2 \times 10^{7}$ iterations with an adaptive step-size targeting an acceptance rate of the order of $45\%$). Figure \ref{FigMicro}(d) compares the estimations of the threshold values $\eta_\alpha$ obtained with MYULA and Px-MALA for different values of $\alpha$, indicating that the approximation errors of MYULA are of the order of $0.1\%$. Moreover, performance analyses based on the chains generated with each algorithm indicate that in this case MYULA is approximately one order of magnitude more computationally efficient than Px-MALA.




\section{Conclusion}\label{sec:conclusion}
This paper presented a new and general proximal MCMC methodology to perform Bayesian computation in log-concave models, with a focus on enabling advanced Bayesian analyses for imaging inverse problems that are convex and not smooth, and currently solved mainly by convex optimisation. The methodology is based on a Moreau-Yoshida-type regularised approximation of the target density that is by construction is log-concave and Lipchitz continuously differentiable, and which can be addressed efficiently by using an unadjusted Langevin MCMC algorithm. We provided a detailed theoretical analysis of this scheme, including asymptotic as well as non-asymptotic convergence results, and bounds on the convergence rate of the chains with explicit dependence on model dimension. In addition to being highly computational efficient and having a strong theoretical underpinning, this new methodology is general and can be applied straightforwardly to most problems solved by proximal optimisation, particularly all problems solved by using forward-backward splitting techniques. The proposed methodology was finally demonstrated with four experiments related to image deconvolution and tomographic reconstruction with total-variation and l1 sparse priors, where we conducted a range of challenging Bayesian analyses related to model comparison and uncertainty quantification, and where we reported estimation accuracy and computational efficiency comparisons with the proximal Metropolis-adjusted Langevin algorithm.

\section{Acknowledgements}
Marcelo Pereyra holds a Marie Curie Intra-European Research Fellowship for Career Development at the School of Mathematics of the University of Bristol,
and is a Visiting Scholar at the School of Mathematical and Computer Sciences of Heriot-Watt University.

\bibliographystyle{plain}
\bibliography{Bibliography/biblio,refs}

\appendix

\section{Proof of \Cref{propo:finite-measure-MY}}
\label{sec:proof-crefpr-meas}
We preface the proof by a Lemma.
\begin{lemma}
  \label{lem:control-fun-convex-gene}
  Let $\gconv: \rset^d \to \ocint{-\infty,\plusinfty}$ be a lower bounded, \lsc~convex function satisfying $    0 < \int_{\rset^d} \rme^{-\gconv(y)}\rmd y < \plusinfty $.
Then there exists $\xgconv \in \rset^d$, $\Rgconv, \rhogconv >0$ such that for all $x \in \rset^d$, $x \not \in \ball{\xgconv}{\Rgconv}$, $\gconv(x)-\gconv(\xgconv) \geq \rhogconv\norm{x-\xgconv}$.
\end{lemma}
\begin{proof}
The proof is a simple extension of the one of \cite[Theorem 2.2.2]{bakry:barthe:cattiaux:guillin:2008}, where $\gconv$ is assumed to be continuously differentiable.

We first show that $\gconv$ is finite on a non-empty open set of
$\rset^d$.  Note since $\int_{\rset^d}
\rme^{-\gconv(y)}\rmd y> 0$, the set $\{ \gconv < \infty \}$ can not be contained in a
$k$-dimensional hyperplane, for $k \in \{0,\cdots,d-1 \}$. Then,
there exists $d+1$ points $\{\rmv_i\}_{0 \leq i \leq d} \subset \{ \gconv < \infty
\}$ such that the vectors $\{\rmv_i-\rmv_0\}_{1 \leq i \leq d}$ are linearly
independent. Denote by $\chull(\rmv_0,\cdots,\rmv_d)$ the convex hull
of $\{\rmv_i\}_{0\leq i \leq d} $ defined by
$$
\chull(\rmv_0,\cdots,\rmv_d)= \left\{ \sum_{i=0}^d \alpha_i \rmv_i \ | \
\sum_{i=0}^d \alpha_i =1 \ , \forall i \in \{0,\cdots,d \} \ , \
\alpha_i \geq 0 \right\} \eqsp.
$$  
Since $\gconv$ is convex and $\chull(\rmv_0,\cdots,\rmv_d) \subset \{ \gconv < \infty \}$, we have
\begin{equation}
    \label{eq:max_conv_hull}
    \sup_{y \in \chull(\rmv_0,\cdots,\rmv_d)} \abs{\gconv(y)}  \leq M_{\chull} = \max_{i \in \{0,\cdots, d\}} \{ \abs{\gconv(\rmv_i)} \} \eqsp.
\end{equation}
It follows from $\{\rmv_i\}_{0 \leq i \leq d} \subset \{ \gconv < \infty \}$ and $\gconv$ is lower bounded that $M_{\chull}$ is
finite. Finally by \cite[Lemma 1.2.1]{florenzano:levan:2001}, $\chull(\rmv_0,\cdots,\rmv_d)$ has non empty interior. 

Consider now the set $\{ \gconv \leq M_{\chull} +1 \}$.  We
prove by contradiction that it is a bounded subset of $\rset^d$. Assume that for all $R
\geq 0$, there exists $x_R \in \{ \gconv \leq M_{\chull} +1 \}$ and
$x_R \not \in \ball{\rmv_0}{R}$. Then since $\{ \gconv \leq M_{\chull} +1 \}$ is
convex, it contains the convex hull of
$\{\rmv_0,\cdots,\rmv_d,x_{R}\}$. Since $\chull(\rmv_0,\cdots,\rmv_d)$
has non empty interior, the volume of
$\chull(\rmv_0,\cdots,\rmv_d,x_{R})$ grows at least linearly in $R$
and the volume corresponding to $\{ \gconv \leq M_{\chull} +1 \}$ is
infinite taking the limit as $R$ goes to $ \infty$. On the other hand,
by assumption and since $\{\rmv_0,\cdots,\rmv_d,x_R\} \subset \{ \gconv \leq
M_{\chull} +1 \}$, we have using the Markov inequality
\begin{equation*}
  \vol\parenthese{  \{ \gconv \leq M_{\chull} +1 \} } \leq  \rme^{M_{\chull}+1} \int_{ \{ \gconv \leq M_{\chull} +1 \}} \rme^{-\gconv(y)} \rmd y < \plusinfty \eqsp,
\end{equation*}
which  leads to a contradiction. Then there exists $\Rgconv \geq 0$, such that $\{\gconv \leq M_{\chull}+1\} \subset \ball{\rmv_0}{\Rgconv}$. \\
For all $x \not \in \ball{\rmv_0}{\Rgconv}$, consider $y =\Rgconv
(x-\rmv_0)\norm[-1]{x-\rmv_0} + \rmv_0$. Note that $y \not \in \{ \gconv
\leq M_{\chull} + 1 \}$, so $\gconv(y) \geq M_{\chull}+1$.  Now using
the convexity of $\gconv$, we have for all $x \not \in
\ball{\rmv_0}{\Rgconv}$,
\begin{equation*}
M_{\chull}+1 \leq \gconv(y) \leq  \Rgconv\norm[-1]{x-\rmv_0} (\gconv(x)-\gconv(\rmv_0)) +\gconv(\rmv_0) \eqsp.
\end{equation*}
Since $\gconv(\rmv_0) \leq M_{\chull}$, we get
\begin{equation*}
  (\gconv(x)-\gconv(\rmv_0)) \geq \Rgconv^{-1}\norm{x-\rmv_0}
\end{equation*}
 and the proof is concluded setting $\xgconv = \rmv_0$.
\end{proof}

\begin{proof}[Proof of \Cref{propo:finite-measure-MY}]
\begin{enumerate}[label=\alph*), wide=0pt, labelindent=\parindent]
\item 
We first assume that \Cref{assum:integrabilite}-\ref{assum:integrable_g} holds.
  By \eqref{eq:id-MY-env}, $U \geq \Uml$ and therefore $0 <
  \int_{\rset^d} \rme^{-U(y)} \rmd y < \int_{\rset^d} \rme^{-\Uml(y)} \rmd y$. We now prove $\rme^{-\gUl}$
  is integrable with respect to the Lebesgue measure, which  implies
  $y \mapsto \rme^{-\Uml(y)}$ is integrable as well since $\fU$ is assumed to be lower bounded.  By \Cref{assum:form-potential} and
  \Cref{lem:control-fun-convex-gene}, there exist $ \rhoUconv >0$, $\xUconv \in \rset^d$ and $M_1 \in \rset$ such that for all $x \in \rset^d$, $\gU(x)-\gU(\xUconv) \geq  M_1 +
  \rhoUconv\norm{x-\xUconv} $. Thus, for all $x \in \rset^d$, we have by \eqref{eq:id-MY-env}
 \begin{align}
\nonumber
\gUl(x) -\gU(x_{\gU})& \geq M_1 + \rhoUconv\norm{\proxgul(x)-\xUconv} +(2 \lambdaMY)^{-1}\norm[2]{x-\proxgul(x)} \\
\label{eq:second bound-finite-measure-MY}
&\geq M_1 + \inf_{y \in \rset^d} \{ \rhoUconv\norm{y-\xUconv} +(2 \lambdaMY)^{-1}\norm[2]{x-y} \}
\geq M_1 +\gconvD^{\lambdaMY}(x) \eqsp,
\end{align}
where $\gconvD^{\lambdaMY}(x)$ is the $\lambdaMY$-Moreau Yosida
envelope of $\gconvD(x) = \rhoUconv \norm{x-\xUconv}$. By
\cite[Section 6.5.1]{Parikh2013}, the proximal operator associated
with the norm is the block soft thresholding given for all $\lambdaMY
>0$ and $x \in \rset^d \setminus \{ 0 \}$ by $\proxglD(x) =
\max(0,1-\lambdaMY/\norm{x}) x $ and $\proxglD(0) =0$. Therefore using
again \eqref{eq:id-MY-env}, it follows that there exists $M_2 \in \rset$ such
that for all $x \in \rset^d$,
\begin{equation*}
\gconvD^{\lambdaMY}(x) \geq  \rhoUconv \norm{x-\xUconv} + M_2 \eqsp.
\end{equation*}
Combining this inequality with \eqref{eq:second bound-finite-measure-MY} concludes the proof.

We now assume that \Cref{assum:integrabilite}-\ref{assum:lipschitz_g}
holds. First, we show that for all $\lambdaMY >0$
\begin{equation}
\label{eq:unif_prox}
\sup_{x \in \rset^d} \{\gU(x) - \gUl(x) \}
\leq \lambdaMY \norm[2][\Lip]{\gU}/2 \eqsp,
\end{equation}
 which will conclude the
proof since $\int_{\rset^d} \rme^{-U(x)} \rmd x < \plusinfty$. Using that $\gU$ is Lipschitz, we have by \eqref{eq:id-MY-env},
for all $x \in \rset^d$
\begin{align*}
  \gU(x) - \gUl(x)&= \gU(x) - \inf_{y \in \rset^d} \defEns{\gU(y) + (2 \lambdaMY)^{-1}\norm[2]{x-y}} = \sup_{y \in \rset^d} \defEns{ \gU(x) - \gU(y) - (2 \lambdaMY)^{-1}\norm[2]{x-y}} \\
&\leq \sup_{y \in \rset^d} \defEns{ \norm[][\Lip]{\gU} \norm{x-y} - (2 \lambdaMY)^{-1}\norm[2]{x-y}} \leq  \lambdaMY \norm[2][\Lip]{\gU}/2 \eqsp,
 \end{align*}
where we have used that the maximum of $u \mapsto au - b u^2$, for $a,b \geq 0$, is given by $a^2/(4b)$.
\item This point is a straightforward consequence of \eqref{eq:definition-grad-prox} and \eqref{eq:lip_moreau_yosida}.
\item   
Since $\pi$  has also a density with
  respect to the Lebesgue measure and $ \Uml(x) \leq U(x)$ for all $x \in \rset^d$, we have for all $\lambda >0$
\begin{equation}
\label{eq:bound_2_TV_MY}
      \tvnorm{\piml-\pi}  =
\int_{\rset^d} \abs{\piml(x)-\pi(x)} \rmd x
\leq 2 A_{\lambdaMY} \eqsp,
\end{equation}
where $A_{\lambdaMY} = \int_{\rset^d}  \{1-\rme^{\gUl(x)-\gU(x)} \} \piml(x)\rmd x  = 1- \defEns{\int_{\rset^d}\rme^{-\Uml(x)} \rmd x}^{-1}\int_{\rset^d} \rme^{-U(x)}  \rmd x$.
By \eqref{eq:limit-d-lambda}, for all $x \in \rset^d$, we get $\lim_{\lambdaMY \downarrow  0} \uparrow  \Uml(x)
= U(x)$. We conclude by applying the monotone convergence theorem.
\item
Using that for all $x \in \rset^d$, $\gUl(x) \leq \gU(x)$ and  $1-\rme^{-u} \leq u$ for all $u \geq 0$, \eqref{eq:bound_2_TV_MY} shows that
\begin{equation*}
\tvnorm{\piml-\pi} \leq 2\int_{\rset^d}\{ \gU(x)-\gUl(x) \} \piml(x) \rmd x \eqsp.
 \end{equation*}
Then the proof follows from \eqref{eq:unif_prox}. 

\end{enumerate}
\end{proof}

\section{Model selection using improper priors}
\label{sec:selection_model_case-proper-imp_prior} 
Model selection using improper priors can lead to tedious
considerations \cite{cprbayes}. Indeed, in that case the joint density
of each model is not defined. However, this difficulty can be avoided
when the considered models share the same improper prior distribution
see \cite{marin:robert:2007:bayesian}. Let $\mathcal{M}_1, \ldots,
\mathcal{M}_K$ be $K$ alternative Bayesian models having the same
improper distribution with density $\tilde{p}(x)$ on $\rset^d$ and
associated to the family of likelihood functions $p_i(y | x)$ such
that for all $i \in \{1,\ldots,K\}$, $\int_{\rset^d} p_i(y|x)
\tilde{p}(x) \rmd x < \plusinfty$. The marginal posterior
probabilities of $\mathcal{M}_1,\ldots, \mathcal{M}_K$ are then
defined by
\begin{equation}\label{margPost}
\tilde{p}(\mathcal{M}_j | y) = \frac{ \tilde{p}(y|\mathcal{M}_j) K^{-1}}{\sum_{k = 1}^K \tilde{p}(y | \mathcal{M}_k) K^{-1}},\quad j \in \{1, \ldots, K\}\, ,
\end{equation}
where for all $j \in \{1,\ldots,K\}$,
\begin{equation*}
\tilde{p}  (y|\mathcal{M}_j)  = \int_{\rset^d} p_i(y|x)
\tilde{p}(x) \rmd x \eqsp.
\end{equation*}

\section{Truncated harmonic mean estimator}\label{HME}

\subsection{Case of proper prior distributions}
\label{sec:case-proper-prior}
Consider a positive probability density $p$ on $\rset^d \times
\rset^m$ for $d,m \in \nset^*$ of the form: $p(x,y) = f(x,y) /
\int_{\rset^d \times \rset^m} f(z,w) \rmd z \rmd w$. Assume that $f$
is known but not the normalization constant of $p$.  Here $p$ plays
the role of a joint distribution of the data and the parameters. It
can be defined if we take a proper prior distribution for the
parameters.
Define for any bounded Borel set $\borelean \in \borelSet(\rset^d)$
\begin{align}
\label{harmonicmean}
  I_{\borelean}(f,y) &= \int_{\rset^d} \1_{\borelean}(x) \frac{p(x|y)}{f(x,y)} \rmd x \\
\nonumber
& = \left. \int_{\rset^d} \1_{\borelean}(x) \frac{p(x|y)}{ p(x,y)} \rmd x  \middle/   \int_{\rset^d \times \rset^{m}} f(z,w) \rmd z \rmd w \right.\eqsp.
\end{align}
Since $p(x|y) = p(x,y)/p(y)$, the following identity holds
\begin{equation}
\label{eq:relation_harmonic_mean}
 p(y) = \left.   \vol(\borelean) \defEns{I_{\borelean}(f,y) \int_{\rset^d \times \rset^{m}} f(z,w) \rmd z \rmd w }^{-1}  \right. \eqsp.
\end{equation}
For all $y \in \rset^m$ and $\borelean\in \mathcal{B}(\rset^d)$, we define the truncated harmonic mean estimator of $I_{\borelean}(f,y)$ by 
\begin{equation}\label{harmonicmean_est_1}
\hat{I}_{\borelean}(f,y)  = \sum_{k = 1}^n \frac{\indi{\borelean}(X_k)}{f(X_k,y )} \eqsp,
\end{equation}
where $(X_k)_{k \geq 1}$ is an ergodic Markov chain targeting
$p(x|y)$ to ensure that the defined estimator almost surely converges to
$I_{\borelean}(f,y)$ given by \eqref{harmonicmean}.

Let $p_1 , p_2$ be two positive distributions on $\rset^d \times
\rset^m$, associated with their two unormalized versions $f_1,f_2:
\rset^d \times \rset^m \to \rset_+$.  We aim to estimate
$p_1(y)/p_2(y)$.  By \eqref{eq:relation_harmonic_mean}, we have
\begin{equation*}
\label{eq:ratio_marginal_likelihood}
  \frac{p_1(y)}{p_2(y)} = 
\frac{\int_{\rset^d \times \rset^{m}} f_2(z,w) \rmd z \rmd w}{\int_{\rset^d \times \rset^{m}} f_1(z,w) \rmd z \rmd w} \frac{I_{\borelean}(f_2,y) } {I_{\borelean}(f_1,y) }
\end{equation*}
Using  \eqref{harmonicmean_est_1}, we estimate this ratio  by
\begin{equation*}
\frac{p_1(y)}{p_2(y)} \approx 
\ratioEst_{1,2}(y) = \frac{\int_{\rset^d \times \rset^{m}} f_2(z,w) \rmd z \rmd w}{\int_{\rset^d \times \rset^{m}} f_1(z,w) \rmd z \rmd w} \frac{\hat{I}_{\borelean}(f_2,y) } {\hat{I}_{\borelean}(f_1,y) }\eqsp.  
\end{equation*}
However, we need to compute the ratio $\int_{\rset^d \times \rset^{m}} f_2(z,w) \rmd z \rmd w /\int_{\rset^d \times \rset^{m}} f_1(z,w) \rmd z \rmd w$, if it not equal to $1$. 

Assume that for $i=1,2$, $f_i(x,y) = h_i(x,y)g_i(x)$, for some measurable functions $h_i: \rset^d \times \rset^m \to \rset^*_+$, $g_i : \rset^d \to \rset^*_+$ such that  $\int_{\rset^m} h_i(x,y) \rmd y $ does not depend on $x$. Note that this assumption holds in \Cref{sec:experiment-2:-image}.
We distinguish two cases:
\begin{enumerate}
\item
If for $i=1,2$,  $g_i$ is integrable, we get
\begin{equation*}
  \ratioEst_{1,2}(y) = \frac{\int_{\rset^d} g_2(z) \rmd z }{\int_{\rset^d} g_1(z) \rmd z} \frac{\hat{I}_{\borelean}(f_2) } {\hat{I}_{\borelean}(f_1) } \eqsp.
\end{equation*}
In the case where the ratio $\int_{\rset^d} g_2(z) \rmd z /
\int_{\rset^d} g_1(z) \rmd z$ is unknown, such as with the priors considered in the experiment reported in Section \ref{sec:experiment-2:-image}, we use a Monte
Carlo algorithm such as MYULA or Px-MALA to compute it. Observe that this computation can be performed offline when the ratio does not depend on the value of $y$.

\item If there exists a function $g: \rset^d \to \rset_+^*$ and two real
  numbers $\homogeneousFactorD_1,\homogeneousFactorD_2 >0$ such that
  for $i=1,2$, $g_i( x) = g(\homogeneousFactorD_i x)$ for all $x \in
  \rset^d$, we get for all $R > 0$
  \begin{multline*}
    \int_{\rset^d \times \rset^{m}} \1_{\ball{0}{R}} f_2(z,w) \rmd z \rmd w /\int_{\rset^d \times \rset^{m}} \1_{\ball{0}{\homogeneousFactorD_1 \homogeneousFactorD_2^{-1}R}} f_1(z,w) \rmd z \rmd w \\
=
\int_{\rset^d } \1_{\ball{0}{R}} g_2(z) \rmd z  /\int_{\rset^d} \1_{\ball{0}{\homogeneousFactorD_1 \homogeneousFactorD_2^{-1}R}} g_1(z) \rmd z  
 = (\homogeneousFactorD_1/ \homogeneousFactorD_2)^d 
\eqsp.
  \end{multline*}
Since for all $a >0$ and $i=1,2$,
  \begin{equation*}
 \int_{\rset^d \times \rset^{m}} f_i(z,w) \rmd z \rmd w  = \lim_{R \to \plusinfty} \int_{\rset^d \times \rset^{m}} \1_{\ball{0}{aR}} f_i(z,w) \rmd z \rmd w 
 \eqsp,
  \end{equation*}
we get 
\begin{equation*}
\left. \int_{\rset^d \times \rset^{m}} f_2(z,w) \rmd z \rmd w \middle/\int_{\rset^d \times \rset^{m}} f_1(z,w) \rmd z \rmd w \right.  =   (\homogeneousFactorD_1/ \homogeneousFactorD_2)^d  \eqsp.
\end{equation*}

\end{enumerate}

\subsection{Case of improper prior distributions}
\label{sec:case-proper-imp_prior}
Let $f : \rset^d \times \rset^m \to \rset_+$ such that for all $y \in
\rset^m$, 
\begin{equation}
  \label{eq:def_marginal_improper}
\tilde{p}(y) = \int_{\rset^d} f(x,y) \rmd x < \plusinfty \eqsp.  
\end{equation}
Here, $f$ plays the role of an improper joint density of the data and
the parameters as the prior distribution is improper. This setting corresponds to \Cref{ssec:exp1}. Define for
all $y \in \rset^m$ the conditional distribution on $\rset^d \times
\rset^m$ by $p(x|y) = f(x,y)/\tilde{p}(y)$, where $\tilde{p}$ is
defined by \eqref{eq:def_marginal_improper}. Then, define 
for any bounded Borel set $\borelean \in \borelSet(\rset^d)$
\begin{equation}
\label{harmonicmean_improper}
  I_{\borelean}(f,y) = \int_{\rset^d} \1_{\borelean}(x) \frac{p(x|y)}{f(x,y)} \rmd x \eqsp.
\end{equation}
Then by \eqref{eq:def_marginal_improper}, we get 
\begin{equation}
\label{eq:relation_harmonic_mean_improper}
 \tilde{p}(y) = \left.   \vol(\borelean) /I_{\borelean}(f,y)  \right. \eqsp.
\end{equation}
For all $y \in \rset^m$ and $\borelean\in \mathcal{B}(\rset^d)$, we define the truncated harmonic mean estimator of $I_{\borelean}(f,y)$ as in \Cref{sec:case-proper-prior} by 
\eqref{harmonicmean_est_1}.

Let now $f_1,f_2: \rset^d \times \rset^m \to \rset_+$, satisfying for
all $i=1,2$ and $y \in \rset^m$, $\tilde{p}_i(y) = \int_{\rset^d}f_i(x,y) \rmd
x < \plusinfty$.  We aim to estimate $\tilde{p}_1(y)/\tilde{p}_2(y)$. But by \eqref{eq:relation_harmonic_mean_improper}, we have 
\begin{equation*}
\label{eq:ratio_marginal_likelihood}
\frac{\tilde{p}_1(y)}{\tilde{p}_2(y)} = 
 \frac{I_{\borelean}(f_2,y) } {I_{\borelean}(f_1,y) } \eqsp.
\end{equation*}
Using \eqref{harmonicmean_improper} and \eqref{harmonicmean_est_1}, we estimate this ratio  by
\begin{equation*}
\frac{\tilde{p}_1(y)}{\tilde{p}_2(y)} \approx 
\ratioEst_{1,2}(y) = \frac{\hat{I}_{\borelean}(f_2,y) } {\hat{I}_{\borelean}(f_1,y) }\eqsp.  
\end{equation*}



\end{document}